\documentclass[10pt]{IEEEtran}
\usepackage{amsmath,amssymb,amsfonts,amsthm}
\usepackage{graphicx}
\usepackage[subrefformat=parens,labelformat=parens]{subcaption}
\usepackage{xcolor}
\usepackage{hyperref}
\usepackage{booktabs}
\usepackage{siunitx}
\usepackage{nccmath}
\usepackage{tikz}
\usepackage{algorithm}
\usepackage{algpseudocode}
\usepackage{makecell}
\usepackage{multirow}
\usepackage{cleveref}
\usepackage[normalem]{ulem}
\usepackage{soul}
\newtheorem{example}{Example}
\newtheorem{definition}{Definition}
\newtheorem{proposition}{Proposition}
\newtheorem{theorem}{Theorem}

\newtheorem{lemma}{Lemma}
\begin{document}
\newif\ifcomments

\newif\ifnotexclude

\commentstrue

\raggedbottom

\newcommand{\claudio}[1]{%
  \ifcomments
  \textcolor{red}{\scriptsize{Claudio: #1}}
  \else
  \fi
}
\newcommand{\valdemar}[1]{%
  \ifcomments
  \textcolor{olive}{\scriptsize{Valdemar: #1}}
  \else
  \fi
}

\newcommand{\daniel}[1]{
  \ifcomments
  \textcolor{blue}{\scriptsize{Daniel: #1}}
  \else
  \fi
}

\newcommand{\review}[1]{
  \ifcomments {\color{red} #1} \else #1 \fi
}

\newcommand{\reviewrm}[1]{
  \ifcomments
  {\color{blue}{\sout{#1}}}
  \else
  \fi
}

\title{Precision on Demand: Propositional Logic for Event-Trigger Threshold Regulation}
\author{Valdemar Tang,
Claudio Gomes and
Daniel E. Lucani \thanks{Manuscript created April, 2024; Revised August 2024; Valdemar Tang (valdemar.tang@ece.au.dk), Claudio Gomes and Daniel Lucani are with the Department of Electrical and Computer Engineering, Aarhus University, 8200 Aarhus, Denmark. 

Copyright © 2024 IEEE. Personal use of this material is permitted. However, permission to use this material for any other purposes must be obtained from the IEEE by sending a request to pubs-permissions@ieee.org.
}}
\maketitle              
\begin{abstract}
We introduce a novel event-trigger threshold (ETT) regulation mechanism based on the quantitative semantics of propositional logic (PL). We exploit the expressiveness of the PL vocabulary to deliver a precise and flexible specification of ETT regulation based on system requirements and properties. Additionally, we present a modified ETT regulation mechanism that provides formal guarantees for satisfaction/violation detection of arbitrary PL properties. To validate our proposed method, we consider a convoy of vehicles in an adaptive cruise control scenario. In this scenario, the PL operators are used to encode safety properties and the ETTs are regulated accordingly, e.g., if our safety metric is high there can be a higher ETT threshold, while a smaller threshold is used when the system is approaching unsafe conditions. Under ideal ETT regulation conditions in this safety scenario, we show that reductions between 41.8 - 96.3\% in the number of triggered events is possible compared to using a constant ETT while maintaining similar safety conditions.
\end{abstract}
\begin{IEEEkeywords}
Propositional Logic, Event-triggering mechanisms, Interval Arithmetic
\end{IEEEkeywords}

\newcommand{\propstates}[2]{\boldsymbol{\hat{x}_{#2}}}
\newcommand{\propstate}[2]{\hat{x}_{#1, #2}}

\newcommand{\ett}[4]{\delta_{#1,#2,#3}(#4, \epsilon_{#1, #2})}
\newcommand{\ettfinal}[2]{\delta_{#1, #2}(\propstates{}{}, \allprops)}
\newcommand{\ettpositive}[4]{\delta^+_{#1, #2,#3}(\propstates{#2}{#4}, \epsilon_{#1, #2})}
\newcommand{\ettpositivenoargs}[3]{\delta^+_{#1,#2,#3}(\cdot)}
\newcommand{\ettrefined}[5]{\bar{\delta}_{#1,#2,#3}(\propstates{#2}{#4}, #5)}
\newcommand{\ettrefinednoargs}[3]{\bar{\delta}_{#1, #2, #3}(\cdot)}
\newcommand{\ettin}[5]{\delta_{#1,#2,#3}^{\in}(\propstates{#2}{#4}, \epsilon_{#1, #2}, #5)}
\newcommand{\ettinwc}[5]{\underline{\delta}_{#1,#2,#3}^{\in}(\Delta\propstates{#2}{#4}, \epsilon_{#1, #2}, #5)}
\newcommand{\ettrefinedarb}[5]{\tilde{\delta}_{#1,#2,#3}(\propstates{#2}{#4}, #5)}
\newcommand{\ettrefinedarbnoargs}[3]{\tilde{\delta}_{#1,#2,#3}(\cdot)}
\newcommand{\etttemprefinement}[6]{\theta_{#1,#2,#3}(\propstates{#2}{#4}, \epsilon_{#1, #2}, #5, #6)}
\newcommand{\etttemprefinementmin}[4]{\bar{\theta}_{#1,#2,#3}(\propstates{#2}{#4}, \epsilon_{#1, #2})}
\newcommand{\etttempset}[3]{\Theta_{#1,#2}(#3)}
\newcommand{\ettarb}[4]{\delta^*_{#1,#2,#4}(\propstates{}{t_k}, #3)}
\newcommand{\ettarbwc}[4]{\underline{\delta}^*_{#1,#2,#4}(\Delta\propstates{}{t_{k+1}|t_k}, #3)}
\newcommand{\ettarbwcbetazero}[4]{\underline{\delta}^*_{#1,#2,#4}(\Delta\propstates{}{t_{k+1}|t_k}, 0)}
\newcommand{\ettmax}[2]{\delta^{max}_{#1, #2}}
\newcommand{\ettnoargs}[3]{\delta_{#1,#2,#3}(\cdot)}
\newcommand{\triggerinstant}[2]{\tau^{#1}_{#2}}
\newcommand{\thmett}[2]{\underline{\delta}^+_{#1, #2, t_{k+1}}(#2, \Delta \propstates{}{t_{k+1} | t_k})}
\newcommand{\rhoett}[1]{\rho ETT}
\newcommand{\rhoettwc}[1]{\underline{\rho}ETT}

\newcommand{\robustness}[2]{\rho(#1, \propstates{#1}{#2})}
\newcommand{\rosil}[3]{\underline{RoSI}(#1, \propstates{#1}{#2}, #3)}
\newcommand{\relativerobustness}[2]{\zeta_\rho(#1, \propstates{#1}{#2})}
\newcommand{\relativerobustnessnextstep}[2]{\underline{\zeta}_\rho(#1, \Delta\propstates{#1}{#2})}
\newcommand{\relativerobustnessrefined}[2]{\tilde{\zeta}_\rho(#1,\propstates{#1}{#2})}
\newcommand{\relativerobustnessrefinednextstep}[2]{\tilde{\underline{\zeta}}_\rho(#1,\Delta\propstates{#1}{#2})}
\newcommand{\relativerobustnessarb}[2]{\zeta_\rho^{*}(#1, \propstates{#1}{#2})}
\newcommand{\relativerobustnessarbnextstep}[2]{\underline{\zeta}_\rho^{*}(#1, \Delta\propstates{#1}{#2})}
\newcommand{\nextsteprobustnesswc}[2]{\underline{\rho}(#1, \Delta \boldsymbol{\hat{x}_{#2}})}
\newcommand{\nextsteprobustnessbc}[2]{\overline{\rho}(#1, \Delta \boldsymbol{\hat{x}_{#2}})}
\newcommand{\robustnessintervals}[2]{\rho(#1, \Delta \boldsymbol{\hat{x}_{#2}})}
\newcommand{\truerobustness}[2]{\rho(#1, \boldsymbol{x}_{#2})}

\newcommand{\propsignals}[1]{\boldsymbol{y_{\sim #1}}}
\newcommand{\allprops}[0]{\boldsymbol{\varPhi}}

\newcommand{\satisfactiontime}[2]{t_s(#1, #2)}
\newcommand{\until}[2]{#1\mathcal{U}_{[t_1, t_2]}#2}
\newcommand{\always}[1]{\square_{[t_1, t_2]}#1}
\newcommand{\nextop}[1]{\bigcirc#1}
\newcommand{\alwaysunb}[0]{\square}
\newcommand{\eventually}[1]{\diamond_{[t_1, t_2]}#1}
\newcommand{\eventuallytime}[2]{\diamond_{[#2]}#1}
\newcommand{\inequalityprop}[1]{\varphi^p_{#1}}

\newcommand{\epsilonparam}[2]{\epsilon_{#1, #2}}
\newcommand{\lambdaparam}[2]{\lambda_{#1, #2}}
\section{Introduction}

Cyber-physical systems can collect large amounts of sensor data making them costly to operate over time. In such systems, data is often transmitted periodically. This causes a consistent overhead of operating the system even when no new or interesting events are happening. Freeing up communication resources from sensors that do not require high accuracy would also allow the system to increase communication with sensors measuring critical values that require high accuracy. The criticality associated with each sensor changes in time and will depend on specific system properties and requirements. Thus, there is a clear motivation to develop methods to match data traffic to the safety and performance requirements of cyber-physical systems to avoid unnecessary resource consumption, delay or congestion in the communication network. 

\paragraph*{Event-triggering mechanism}
A popular mechanism to implement efficient communication in the state-of-the-art (SOTA) is event-triggered state estimation and control. The fundamental idea is to only sample and/or transmit information when necessary. The necessity is defined by an event-triggering condition that must be fulfilled for the system to make use of the measurement. This can help to significantly reduce communication costs compared to periodic sampling and transmission \cite{chenEventTriggeredStateEstimation2017}. The event-triggering condition typically uses an event-triggering threshold (ETT) which indicates when the difference between new and known information is large enough to be transmitted. Figure \ref{fig:transmission-strategies} visualizes the conceptual differences between a periodic transmission strategy (top), a static threshold (middle) and a safety-dependent threshold (bottom). Intuitively, smaller thresholds (and thus better accuracy) are beneficial in unsafe situations. A similar parallel can be drawn to performance, where underperforming systems require better accuracy and well-performing systems can potentially operate with less sensor accuracy.

\begin{figure}[h]
    \centering
    \includegraphics[width=0.45\textwidth]{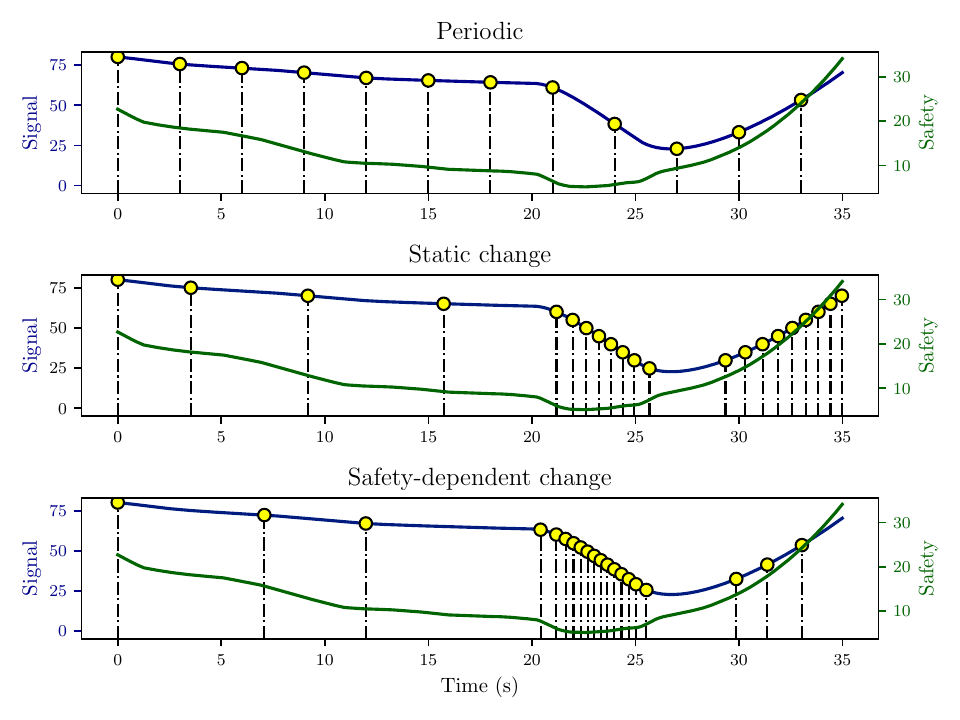}
    \caption{Visualization of periodic transmission (top), static (middle) and safety-dependent (bottom) event-triggering conditions. Yellow dots indicate measurement transmission. The static event trigger prioritizes change in the signal regardless of safety, while the safety-dependent event trigger prioritizes change relative to the safety of the system.}
    \label{fig:transmission-strategies}
\end{figure}
Initially conceptualized in \cite{aarzenSimpleEventbasedPID1999,johanastromComparisonPeriodicEvent1999}, event-triggered techniques have since undergone significant development and have been applied to reduce communication in multi-agent consensus systems \cite{dingOverviewRecentAdvances2018}, networked control systems (NCSs) \cite{zhangNetworkedControlSystems2020a} and distributed state estimation \cite{jiaResourceefficientSecureDistributed2021} among others. As the goals of these systems are different, so are the event-triggering conditions. For example, in NCSs a central focus is on designing a controller and event-triggering mechanism to ensure stability of the controlled systems in the presence of transmission delays, dropouts and other phenomena introduced by the network. As a result, the event-triggering condition is often designed in combination with the controller and the system dynamics. Controller systems may employ state observers (a.k.a. state estimators), which means that using event-triggered techniques requires a corresponding adaptation of those observers. Several well-known techniques have been adapted to event-triggered state estimation, e.g., the Kalman filter \cite{liEventTriggeredKalmanFilter2023}, $H_{\infty}$ filtering \cite{zhangEventbasedFilteringSampleddata2015} and set-membership filtering \cite{weiEventbasedDistributedSetmembership2016}.

\paragraph*{Propositional Logic}
While some properties such as stability and convergence are typically well-defined for a broad class of systems, other safety (and performance) properties require a formal definition on a per-system basis.
Multiple formalisms can be used to achieve this goal. In this paper, we focus on signal-based Propositional Logic (PL) statements \cite{buningPropositionalLogicDeduction1999a} as detailed later in Section \ref{sec:stl}.
PL is fundamentally different from stability analysis in that it provides a way to express and verify system properties based on time series data rather than exhaustively proving system stability.
A PL formula outputs a verdict specifying whether the given property was satisfied or violated. Crucially, it provides a quantitative metric that signifies how much the property is violated or how far it is from being violated. Such properties can be related to performance (liveness) or safety.

\paragraph*{Contribution}
Since the design of the ETT depends on the goals of the system, and since PL provides a well-founded approach to describe and monitor such goals, we propose a method for leveraging the expressiveness of PL to enable precise regulation of the ETT based on system requirements. More specifically, we provide a parameterized ETT regulation mechanism for inequality PL properties and their negations and use the propositional operators and structure of the propositional property to refine the ETT. We provide a method for determining a set of parameters that under certain assumptions guarantee achieving the best possible accuracy before violating a safety or performance property. Additionally, we also provide guidance on parameter tuning for the proposed ETT regulation mechanism and evaluate our approach on a case study concerning a convoy of vehicles in an adaptive cruise control scenario. We compare the numerical results in terms of the number of triggered events for the time-triggered case, a constant ETT and our proposed ETT regulation mechanism. Finally, we explore key tradeoffs for different parameter configurations. 

\paragraph*{Structure}
The considered system architecture, background terminology and problem formulation are presented in Section \ref{sec:background} followed by the proposed method in Section \ref{sec:stl-properties-ett}. This method is applied to a simulated adaptive cruise control (ACC) scenario in Section \ref{sec:case-study}. We discuss the results and limitations of our method in Section \ref{sec:discussion} as well as compare our work with SOTA. Finally, we present concluding remarks and planned future work in Section \ref{sec:conclusion}.
\section{Background}\label{sec:background}
\renewcommand{\arraystretch}{1.3}
\begin{table}[h] 
    \centering 
    \begin{tabular}{|c|p{0.2\textwidth}|} 
        \hline 
        \textbf{Notation} & \textbf{Description} \\ \hline 
        $\boldsymbol{\hat{x}_{t_k}}, \hat{x}_{i,t_k}$ & State estimate vector and state estimate vector element at time $t_k$. \\ \hline
        $\boldsymbol{y}, \propsignals{\inequalityprop{}}$ & Measurable system output vector and signals to regulate ETTs for based on the inequality property $\inequalityprop{}$ respectively. \\ \hline 
        $\varphi, \inequalityprop{}$ & Arbitrary propositional and inequality property respectively \\ \hline
        $\ett{y_i}{\varphi}{t_{k+1}}{t_k}$ & ETT of signal $y_i$ at time $t_k$ defined based on property $\varphi$. \\ \hline 
        $e(\cdot)$ & Update error. Difference between predicted and current or last and current measurement. \\ \hline 
        $\robustness{\varphi}{t_k}$ & Robustness/satisfaction degree of the property $\varphi$ and corresponding system states at time $t_k$. \\ \hline
        $\Delta x, w(\Delta x)$ & Interval and widht of the variable $x$ respectively. \\ \hline
        $\Delta \propstates{}{t_{k+1} | t_k}$ & Vector containing intervals of predicted states at time $t_{k+1}$ based on the states at time $t_k$. \\ \hline
        $\nextsteprobustnesswc{\varphi}{t_{k+1}|t_k}$, $\nextsteprobustnessbc{\varphi}{t_{k+1}|t_k}$ & Lower and upper bound respectively, of the predicted robustness interval of the inequality property $\varphi$ at time $t_{k+1}$ based on the state interval vector. \\ \hline
        $\epsilonparam{y_i}{\varphi}, \lambdaparam{y_i}{\varphi}, \epsilonparam{\rho}{\inequalityprop{}}$ & ETT regulation parameters related to signal $y_i$ and/or inequality property $\inequalityprop{}$. \\ \hline
        $\relativerobustness{\varphi}{t_k}$ & The normalized robustness of the property $\varphi$ at time $t_k$. Variants with different super-scripts exist.  \\ \hline
    \end{tabular} 
    \caption{Overview of common notation used in the paper.} 
    \label{tab:notation} 
\end{table} 

\subsection{The event-triggered system architecture under study}\label{sec:system-architecture}

We consider the system architecture in Fig. \ref{fig:system-architecture} with multiple smart sensors measuring a process. The true state vector of the system at time $t_k$ is denoted by $\boldsymbol{x_{t_k}}$, the control input vector by $\boldsymbol{u_{t_k}}$ and $f(\boldsymbol{x_{t_k}}, \boldsymbol{u_{t_k}})$ is the function that describes how the system evolves from $t_k \rightarrow t_{k+1}$.
Sensor $i$ measures a value $y_{i,t_k}$ at time $t_k$ given by $y_{i, t_k} = w_i(\boldsymbol{x_{t_k}})$. If the event-triggering condition is satisfied, the measurement $y_{i,\tau^i_{m+1}}$ is transmitted to the remote state estimator, where $\tau^i_{m+1}$ identifies the current event-triggering time instant for the $i$th sensor. The state of the process is estimated by the remote state estimator using a model of the process, the previously estimated state $\boldsymbol{\hat{x}_{t_{k-1}}}$, measurements and ETT information from the sensors and monitor $\boldsymbol{\hat{y}_{t_{k-1}}}$ and control input $\boldsymbol{u_{t_{k-1}}}$.  
The system also has a set of properties based on the system requirements $(\boldsymbol{\varPhi})$, which are monitored by the monitor. The monitor receives the $\boldsymbol{\hat{x}_{t_k}}$ and optionally $\boldsymbol{u_{t_k}}$ to determine the state/satisfaction of the monitored properties. Based on the current state of the properties and the measurable system outputs $\boldsymbol{y}$, the monitor calculates suitable ETTs for each smart sensor. The smart sensor sampling and control system update frequency are assumed to be the same and synchronized. Additionally, we consider the case where the system, under the periodic transmission case (i.e. all samples are transmitted) satisfies all system properties.
\begin{figure*}[t]
    \centering
    \includegraphics[width=\textwidth]{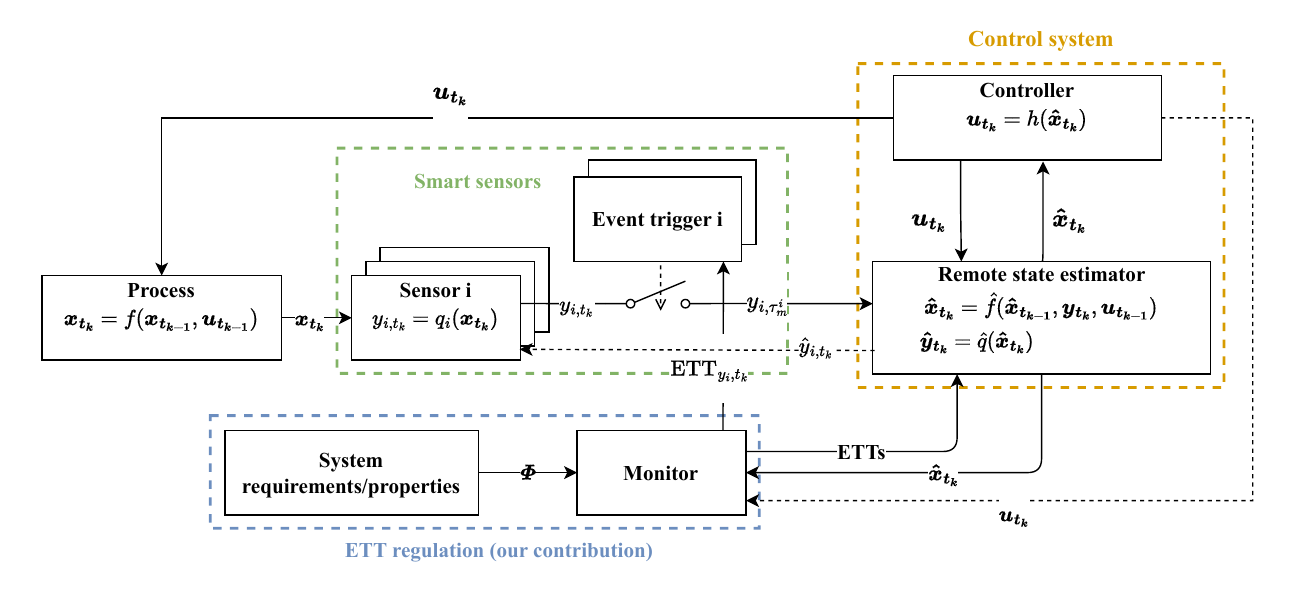}
    \caption{An overview of the system architecture considered in this paper. Dotted black lines indicate an optional data flow. The control input $\boldsymbol{u}_{t_k}$ is necessary for the monitor if we wish to verify or monitor control actions in the system properties and $\hat{y}_{i, t_k}$ is necessary if the innovation update error is used (described in Section \ref{sec:event-triggering-mechanisms}). }
    \label{fig:system-architecture}
\end{figure*}

\subsection{ETT regulation strategies}\label{sec:event-triggering-mechanisms}

The literature distinguishes between static and dynamic ETTs. \textbf{Static ETTs} depend on the current state estimate vector, are time-dependent or are constant, whereas \textbf{dynamic ETTs} have additional dynamic variables. We focus on state-dependent static ETTs and thus do not detail dynamic ETTs further and interested readers are referred to \cite{geDynamicEventTriggeredDistributed2020,geDynamicEventtriggeredControl2021,geDistributedEventTriggeredEstimation2020,jiaResourceefficientSecureDistributed2021,zhangOverviewDeepInvestigation2017}. A general event-triggered transmission mechanism for a state-dependent static ETT can be defined as a sequence of event-triggering times \cite{geDynamicEventtriggeredControl2021} for a sensor $i$:

\begin{equation}\label{eq:static-ett}
    \triggerinstant{i}{m+1} = \inf \lbrace t > \triggerinstant{i}{m} \text{ } | \text{ } g_e(e(\cdot),t, ...) > g_{x}(\boldsymbol{\hat{x}_{t}}, ...)\rbrace,
\end{equation}
where $e(\cdot)$ is the update error and can be defined as the difference between the current measurement and previous transmitted measurement $(y_{i, t_{k}} - y_{i, \tau^i_{m}})$ (also known as the \textbf{Send-on-Delta (SOD)} method \cite{miskowiczSendOnDeltaConceptEventBased2006a}) or the difference between the predicted and actual measurement 
\begin{equation}\label{eq:innovation}
    |\hat{y}_{i, t_{k}} - y_{i, t_{k}}|
\end{equation}
also known as the \textbf{innovation} \cite{chenEventTriggeredStateEstimation2017}.
In the SOTA, state-dependent static ETTs are typically given by $g_e(e(\cdot), \Gamma^{\frac{1}{2}}) = ||\Gamma^{\frac{1}{2}}e(\cdot)||^2$ and $g_x(\boldsymbol{\hat{x}_{t_k}}, \Gamma^{\frac{1}{2}}) = ||\Gamma^{\frac{1}{2}}Z(\boldsymbol{\hat{x}}_{t_k}, ...)||^2$ where $\Gamma$ is a weighting matrix and $Z(\boldsymbol{\hat{x}_{t_k}},...)$ is a suitable function to be determined \cite{geDynamicEventTriggeredDistributed2020}. 
\vspace*{1pt}
\begin{example}[SOD Event-trigger]\label{ex:ett}
    Consider a vehicle in an adaptive cruise control scenario that has an estimate of its position $x^c$ as well as the preceding vehicle $x^{c+1}$ where $x^c, x^{c+1} \in \mathbb{R}$. Intuitively, the further away vehicle $c$ is from vehicle $c+1$, the more uncertainty 
    we can allow. For example, we can define $Z(x^c, x^{c+1}) = \max(x^{c+1} - x^c - 10, 1)$, where the minimum ETT is then obtained once the vehicles are less than 11 meters apart. Now consider the position estimates $\hat{x}^c_{t_k} = 10$, $\hat{x}^{c+1}_{t_k} = 22$, the previous measured distance $y_{\tau^c_m} = 12$ from vehicle $c$ to $c+1$ at time $\tau^c_m$, $\underline{\sigma} = 1$, and the weight $\Gamma^{\frac{1}{2}} = 0.2$. We calculate $Z(\cdot) = 22 - 10 = 12$. A new data sample is now taken at the sensor with the value of $y_{t_k} = 15$. 
    The update error becomes: $e(y_{t_k}, y_{\tau^c_m}) = 15 - 12 = 3$.
    We can now determine if an event is triggered by calculating the event-triggering condition:
    \begin{align*}
        & g_e(e(y_{t_k}, y_{\tau^c_m}), \Gamma^{\frac{1}{2}}) = 0.2 \cdot 3 = 0.6, \\
        & g_x(Z(\hat{x}^c_{t_k}, \hat{x}^{c+1}_{t_k}), \Gamma^{\frac{1}{2}}) = 1 \cdot 0.2 \cdot 2 = 0.4.
    \end{align*}
    As $g_e(\cdot) > g_x(\cdot)$, the event-triggering condition is satisfied and the measurement is transmitted.
\end{example}

A central point in the above example is that the ETT is relative to some function of $\boldsymbol{\hat{x}_{t_k}}$. Thus, $g_x(\boldsymbol{\hat{x}}_{t_k}, ...)$ should be designed such that it produces a small value in situations that require a small ETT. 

Event-triggering conditions can be designed for many different purposes. For example, event-triggered techniques have seen widespread use in multi-agent systems (MASs) \cite{dingOverviewRecentAdvances2018}. In such systems, a key design objective could be to achieve consensus in bounded time, to which time-dependent ETTs have been implemented to reduce the ETT over time as the agents gradually achieve consensus \cite{seybothEventbasedBroadcastingMultiagent2013}. Event-triggering techniques have since been developed for a wide range of system classes and setups from linear MASs \cite{huConsensusLinearMultiAgent2016a} to heterogeneous nonlinear MASs \cite{liObserverBasedEventTriggeredIterative2024}.
Event-triggering mechanisms may also adapt to the state of the network. For example, if the network is congested, we may want to increase the ETT resulting in fewer triggered events to avoid further congestion, e.g., in \cite{geDynamicEventTriggeredScheduling2022}. Alternatively, the event-triggering condition may be designed in combination with a controller to ensure the stability of the controlled system, which is the primary focus in event-triggered networked control systems \cite{zhangNetworkedControlSystems2020}.

\subsection{Event-triggered state estimation (ETSE)}\label{apdx:event-triggered-state-estimation}

The event-triggered transmission scheme requires an adaptation of the classical state estimation techniques for periodic sampling and transmission. This is because the estimator has to handle both point (explicit) and set-valued (implicit) measurements where the latter is the result of the ETT \cite{shiEventBasedStateEstimation2016}. There are two overall approaches for tackling this problem: a deterministic and a stochastic approach, each of which determines how the signals are modeled. A central issue in the stochastic approaches is how to deal with the non-gaussian distribution introduced by the ETT. Therefore, the traditional Kalman filter is not directly applicable. An early attempt was to implement the Kalman filter with intermittent observations \cite{sinopoliKalmanFilteringIntermittent2004a}, where the update step of incorporating the measurement is simply skipped at time instants when measurements are not available. This can however lead to sub-optimal performance compared to approaches that take the knowledge of intermittent observations into account \cite{chenEventTriggeredStateEstimation2017}. Exact or optimal solutions to the event-based filtering problem are available but they require numerically estimating the probability distributions which makes them infeasible to use during runtime \cite{shiEventBasedStateEstimation2016}. Therefore approximate solutions are necessary. For simplicity, we focus on state estimators for discrete-time linear systems of the type:
\begin{equation}\label{eq:lti}
    \begin{split}
        & \boldsymbol{x_{t_{k+1}}} = A\boldsymbol{x_{t_k}} + B\boldsymbol{u_{t_k}} + \boldsymbol{w_{t_k}}, \\
        & \boldsymbol{y_{t_k}} = C \boldsymbol{x_{t_k}} + \boldsymbol{r_{t_k}},
    \end{split}
\end{equation}
where $\boldsymbol{x_{t_k}} \in \mathbb{R}^n$ is the system state vector at time $t_k$, $\boldsymbol{y_{t_k}} \in \mathbb{R}^m$ is the measurable system output at time $t_k$, $\boldsymbol{r_{t_k}} \in \mathbb{R}^m$ and $\boldsymbol{w_{t_k}} \in \mathbb{R}^n$ are the measurement and process noise respectively at time $t_k$ and $A, B$ and $C$ are matrices of appropriate sizes. 

Later in our experiments, we use an event-triggered state estimation method that explicitly utilizes the knowledge of the ETT to update the state estimate. One such method was proposed in \cite{shiEventtriggeredApproachState2014} and was later compared with several other ETSE methods in \cite{chenEventTriggeredStateEstimation2017} where it performed the best overall in terms of state estimation accuracy and resulting communication rate. We compared the method from \cite{shiEventtriggeredApproachState2014} with an adaptation of the Kalman filter originally developed for the SOD method in \cite{suhModifiedKalmanFilter2007} by empirically evaluating the estimation accuracy for the system model, where we found that they produced similar estimation accuracy for the case study scenario. As the modified Kalman filter (KF) from \cite{suhModifiedKalmanFilter2007} is significantly simpler, we use this approach for our experiments.

\begin{definition}[Modified Kalman filter adapted from \cite{suhModifiedKalmanFilter2007}]\label{def:modified-kalman-filter}
    A modified Kalman filter which utilizes the knowledge of the event-triggering threshold is given by the standard Kalman prediction equations: 
    \begin{align}
        & \boldsymbol{\hat{x}_{t_k}^-} = A \boldsymbol{\hat{x}_{t_{k-1}}} + B \boldsymbol{u_{t_{k-1}}}, \\
        & P_{t_k}^- = A P_{t_{k-1}}A^T + Q,
    \end{align}
    a modification of the measurement and measurement covariance matrix $R$:
    \begin{equation}
        \bar{R}_{t_k} = R, \\
    \end{equation}
    then for all signals if measurement of signal $y_i$ is received at time $t_k$:
    \begin{equation}
        \indent y_{i, t_k} = y_{i, t_k} \\
    \end{equation}
    else use the predicted measurement $\hat{y}^-_{i, t_k} = C\boldsymbol{\hat{x}^{-}_{t_k}}$:
    \begin{align}
        y_{i, t_k} = \hat{y}^-_{i, t_k}.
    \end{align}
    
    Then, if a measurement was not received for signal $y_i$, the value of a given measurement is assumed to be equal to the predicted value and have a uniform distribution with variance $\delta_{y_i, t_k}^2/3$ where $\delta_{y_i, t_k}$ is the ETT of signal $y_i$ at time $t_k$ which is then added to the measurement noise covariance matrix at the corresponding entry.
    \begin{equation}
        R_{t_k}(i,i) = R_{t_k}(i,i) + \delta_{y_i, t_k}^2/3. \label{eq:update}
    \end{equation}
    We then proceed with the standard Kalman filter update equations:
    \begin{align}
        & K_{t_k} = P_{t_k}^-C^T(CP_{t_k}^-C^T + \bar{R}_{t_k})^{-1}, \\
        & \boldsymbol{\hat{x}_{t_k}} = \boldsymbol{\hat{x}_{t_k}^-} + K_{t_k}(\boldsymbol{y_{t_k}} - C\boldsymbol{\hat{x}_{t_k}^-}), \\
        & P_{t_k} = (I - K_{t_k}C)P_{t_k}^-,
    \end{align}
    where $R$ is the measurement noise covariance matrix, $Q$ is the process noise covariance matrix and $A, B$ and $C$ are the state-space equations from the system model in Eq. \ref{eq:lti}. 
\end{definition}

\subsection{Propositional Logic}\label{sec:stl}

We consider signal-based PL statements with quantified semantics.
PL is defined by the following recursive syntax
\begin{align}\label{eq:stl-syntax}
   \varphi ::= p(\boldsymbol{\hat{x}}) > c\text{ } |\text{ } \neg \varphi \text{ }|\text{ } \varphi_1 \land \varphi_2 \text{ }|\text{ } \varphi_1 \lor \varphi_2,
\end{align}
where $\boldsymbol{\hat{x}}$ is the state estimate vector, $c$ is a constant,
$\neg$ is the negation operator $\land, \lor$ are the logical \textit{and} and \textit{or} operators respectively. 
The syntax in Eq. \eqref{eq:stl-syntax} can be combined to form the well-known implication operator ($\varphi_1 \rightarrow \varphi_2 \equiv \neg \varphi_1 \lor \varphi_2$) which indicates that $\varphi_1$ implies $\varphi_2$. We introduce the notation $\inequalityprop{}$ to refer to an inequality property as this will become useful later.

We adopt the quantitative metrics (robustness) of propositional properties from Signal Temporal Logic \cite{malerMonitoringTemporalProperties2004} that indicate how robustly a given property is violated/satisfied.

\setlength{\belowdisplayskip}{0pt}
\begin{definition}[Robustness adapted from \cite{malerMonitoringTemporalProperties2004}]\label{def:stl-robustness}
    The robustness $\rho$ of a PL property $\varphi$ based on the system state $\boldsymbol{\hat{x}_{t_k}}$ at time $t_k$, denoted as $\rho(\varphi, \boldsymbol{\hat{x}_{t_k}})$ is defined recursively by:
    \begin{align*}
        \rho(p(\boldsymbol{\hat{x}_{t_k}}) > c) &= p(\boldsymbol{\hat{x}_{t_k}}) - c \\
        \rho(\neg\varphi, \boldsymbol{\hat{x}_{t_k}}) &= -\rho(\varphi, \boldsymbol{\hat{x}_{t_k}}) \\
        \rho(\varphi_1 \land \varphi_2, \boldsymbol{\hat{x}_{t_k}}) & =  \min(\rho(\varphi_1, \boldsymbol{\hat{x}_{t_k}}), \rho(\varphi_2, \boldsymbol{\hat{x}_{t_k}})) \\
        \rho(\varphi_1 \lor \varphi_2, \boldsymbol{\hat{x}_{t_k}})& = \max(\rho(\varphi_1, \boldsymbol{\hat{x}_{t_k}}), \rho(\varphi_2, \boldsymbol{\hat{x}_{t_k}})) \\
    \end{align*}
\end{definition}
\setlength{\belowdisplayskip}{5pt}

\subsection{Interval arithmetic}\label{sec:interval-arithmetic}

Due to the ETT, the values of the system outputs can be expressed as an interval $Y = [\underline{Y}, \overline{Y}]$ in terms of the ETT as \cite   {mooreIntroductionIntervalAnalysis2009a}:
\begin{equation}\label{eq:signal-ett-interval}
    Y(y, \delta_{y}) = [y - \delta_y, y + \delta_y] = [\underline{Y}, \overline{Y}],
\end{equation}
where $\delta_y$ is the ETT of the signal $y$ and $y$ is the previously transmitted value or the estimated value in the case of the SOD update error and innovation-based update error respectively. For the remainder of the paper, we use the notation $\Delta x(\cdot)$ to refer to the interval of $x \in \mathbb{R} $ and $\Delta \boldsymbol{x}(\cdot)$ to refer to a vector containing intervals of $\boldsymbol{x} \in \mathbb{R}^n$. We say that an interval is degenerate if $\underline{Y} = \overline{Y}$, which we assume to be the case when an event is triggered as the true value of the signal is known. 
Furthermore, we introduce the addition of intervals $X$ and $Y$ as \cite{mooreIntroductionIntervalAnalysis2009a}:
$X + Y = [\underline{X} + \underline{Y}, \overline{X} + \overline{Y}]$, the subtraction of intervals $X$ and $Y$ as: $X - Y = [\underline{X} - \overline{Y}, \overline{X} - \underline{Y}]$,
and the width of an interval $Y$ as:
\begin{equation}\label{eq:interval-width}
    w(Y) = \overline{Y} - \underline{Y}
\end{equation}
and the product of two intervals $X$ and $Y$ as
\begin{equation}\label{eq:interval-multiplication}
    X \cdot Y = [\min S, \max S] \text{ where } S = \lbrace {\underline{XY}, \underline{X}\overline{Y}, \overline{X}\underline{Y}, \overline{XY}}\rbrace.
\end{equation}
We now present two results that will be useful later to provide property satisfaction detection guarantees.
\begin{lemma}\label{lem:interval-reversal}
    Multiplying an interval $X = [\underline{X}, \overline{X}]$ where $\underline{X} \neq \overline{X}$ with a scalar $\alpha < 0$, reverses the interval s.t. $\alpha X = [\overline{X}\alpha, \underline{X}\alpha]$.
\end{lemma}
\begin{proof}
    See appendix \ref{apdx:interval-reversal-proof}.
\end{proof}
\begin{lemma}\label{lem:ett-interval-width} The addition or subtraction of two intervals $X(x, \delta_x)$ and $Y(y, \delta_y)$ scaled by scalars $\alpha_X$ and $\alpha_Y$ respectively, where $\alpha_X \in \mathbb{R}, \alpha_Y \in \mathbb{R}$, produces an interval where $w(\alpha_X X(x, \delta_x) \pm \alpha_Y Y(y, \delta_y)) = 2|\alpha_X|\delta_x + 2|\alpha_Y|\delta_y$.
\end{lemma}
\begin{proof}
    See appendix \ref{apdx:ett-interval-width-proof}.
\end{proof}

For more information on interval arithmetic, the reader is referred to \cite{mooreIntroductionIntervalAnalysis2009a}.

\subsection{Problem formulation}

We introduce the notation $\propsignals{\varphi}$ where $\propsignals{\varphi} \subseteq \boldsymbol{y}$ to denote the system outputs to which the ETT regulation should be applied based on the property $\varphi$.
The set $\propsignals{\varphi}$ is to be defined by the system designer.

We consider the problem of finding an event-triggering mechanism $ET(\Phi, \propstates{}{}, \Xi)$ with a corresponding set of parameters $\Xi$ to ensure the satisfaction of all system properties $\boldsymbol{\varPhi}$ while transmitting as few measurements as possible. This is stated formally as the following optimization problem:
\begin{equation}\label{eq:general-optimization}
    \begin{split}
        & \underset{ET(\Phi, \propstates{}{}, \Xi)}{\text{argmin}}(\sum_{y_i \in \boldsymbol{y}} |\lbrace y_{i, \tau^i_1}, ..., y_{i, \tau^i_m}\rbrace |) \\
        & \indent \indent \text{s.t. } \forall \varphi \in \boldsymbol{\varPhi} \text{ } \forall t_k \in [t_0, t_{end}] \text{ } \rho(\varphi, \propstates{}{t_k}) > 0, \\
    \end{split}
\end{equation}

Additionally, we assume that the original system under the time-triggered transmission approach (TT) satisfies all system properties. Thus we know that
\begin{equation}\label{eq:ett-feasibility}
    \begin{split}
        & ET(\cdot) = TT(T_s) \rightarrow \\
        & \indent \indent (\rho(\varphi, \propstates{}{t_k}) > 0) \text{ }  \forall t_k \in [t_0, t_{end}] \forall \varphi \in \boldsymbol{\varPhi},
    \end{split}
\end{equation}
where $T_s$ denotes the sampling, transmission and controller update interval.

\section{PL robustness-based ETT regulation}\label{sec:stl-properties-ett}

An ETT based on state information requires the designer to model $g_x(\boldsymbol{\hat{x}_t}, ...)$ from Eq. \eqref{eq:static-ett} so that it outputs a small value in situations with a low desired state estimate uncertainty and vice versa. 
The robustness output of inequality  PL properties provides exactly this. For example, we may allow a larger state-estimate uncertainty (and thus save communication resources) in definitely safe or well-performing states (large robustness) and require a lower uncertainty in unsafe/underperforming states (low robustness) as depicted in Fig. \ref{fig:transmission-strategies}. 

\subsubsection*{Relation with State of the Art (SOTA)} The robustness of inequality PL properties and their negations in Definition \ref{def:stl-robustness} is calculated by subtracting some state-dependent value from a constant. 
This is similar to several implementations of $g_x(\boldsymbol{\hat{x}}, ...)$ in SOTA. 
For example, \cite{geDynamicEventTriggeredScheduling2022} proposes a triggering mechanism that depends on the difference between the steady-state desired distance and the current distance between automatically controlled vehicles in a cooperative cruise-control scenario, which is similar to Example \ref{ex:ett}.

In \cite{huConsensusLinearMultiAgent2016}, the authors study a multi-agent system that attempts to achieve consensus where the triggering condition is relative to a weighted difference of state estimates between agents. The smaller the difference, the smaller the ETT. 
This is analogous to a PL liveness property where we want to verify that the system achieves consensus. This similarity indicates that the quantitative semantics of PL properties is a suitable metric to base our ETT regulation mechanism on.

\subsubsection*{Section overview}Initially we define the baseline constant ETT policy in Section \ref{sec:constant-ett}. Then, in the following sections, we gradually build and extend our ETT regulation mechanism from inequality properties in Section \ref{sec:basic-regulation-mechanism} to include PL properties with inequality sub-properties in Section \ref{sec:ett-regulation-propositional} and then arbitrary PL properties in Section \ref{sec:ett-regulation-arb-propositional}. Along the way, we outline the parameter selection procedure in Section \ref{sec:param-selection-procedure} and provide guarantees for detecting property satisfaction first for inequality properties in Section \ref{sec:parameter-tuning} which we then extend to arbitrary PL properties in Section \ref{sec:arb-prop-detection-guarantees}.

\subsection{Baseline scheme: Constant ETTs (CETT)}\label{sec:constant-ett}

We consider the constant ETT policy CETT as a baseline to compare our proposed method to. 
We adapt Eq. \eqref{eq:general-optimization} where we consider the problem of finding a set of constant ETTs that minimize the total number of triggered events, while still ensuring that the system satisfies all properties. Thus under the CETT scheme $\Xi = \lbrace \delta_{y_0}, ..., \delta_{y_i}\rbrace$ and we add the constraint $\forall y_i \in \boldsymbol{y} \text{ } \delta_{y_i} \geq 0$.
The constant ETTs are chosen such that they can ensure property satisfaction under the worst possible system conditions, and thus do not take advantage of less strict accuracy requirements in less critical situations. Thus, if we required a similar accuracy in the bottom plot of Fig. \ref{fig:transmission-strategies}, we would need to set the ETT correspondingly resulting in significantly more triggered events in less critical situations.

\subsection{A robustness-proportional ETT regulation mechanism for inequality properties}\label{sec:basic-regulation-mechanism}

In order to mitigate the disadvantages of the constant ETT policy, we propose a runtime ETT regulation mechanism based on the robustness of inequality PL properties that can take advantage of less critical situations.

\begin{definition}\label{def:ett-regulation-safety-1}
    A static robustness-proportional ETT regulation mechanism for a signal $y_i \in \propsignals{\varphi}$ for the inequality property $\inequalityprop{}$ of one of the types: $p(\boldsymbol{\hat{x}_{\varphi}}) > c, (p(\boldsymbol{\hat{x}_\varphi}) < c)$ is defined as:
    \begin{equation}\label{eq:ett-regulation-safety-1}
        \ett{y_i}{\varphi}{t_{k+1}}{\propstates{}{t_k}} = \frac{\robustness{\varphi}{t_k}}{\epsilonparam{y_i}{\varphi}}
    \end{equation}
    where $\epsilonparam{y_i}{\varphi} \in \mathbb{R}^+$.
\end{definition}

We can now adapt the optimization problem in Eq. \ref{eq:general-optimization} such that $\Xi$ becomes the set of defined $\epsilonparam{y_i}{\inequalityprop{}}$ parameters and we add the constraint that $\epsilonparam{y_i}{\inequalityprop{}} \in \mathbb{R}^+$.

\begin{proposition}\label{prop:ett}
    Using the ETT regulation mechanism in Definition \ref{def:ett-regulation-safety-1}, we achieve the property
    \begin{equation}
        \lim_{\robustness{\varphi}{t_k} \rightarrow 0} \ett{y_i}{\varphi}{t_{k+1}}{\propstates{}{t_k}} = 0
    \end{equation}\label{eq:ett-property}
    which corresponds to smaller ETTs and thus improved accuracy in more critical situations. 
\end{proposition}
For the remainder of this Section, we will focus on the ETT regulation mechanism in Definition \ref{def:ett-regulation-safety-1} and the extension of this to arbitrary PL properties. We refer to this policy as $\rho$ETT. Later, in Section \ref{sec:case-study}, we will revisit the TT approach in Eq. \eqref{eq:ett-feasibility} and the CETT policy and compare these approaches to the $\rho$ETT policy presented in this Section. Next, we sketch the overall potential optimization approaches for determining suitable parameters for the TT, CETT and $\rho$ETT policies.

\subsection{Parameter selection procedure outline}\label{sec:param-selection-procedure}

Determining the $\epsilonparam{y_i}{\inequalityprop{}}$ parameters or constant ETTs in Eq. \ref{eq:general-optimization}, can be categorized into two overall approaches; empirical and formal. The former approach entails running a large number of simulations for different parameter sets with random measurement noise and starting conditions. For example, if we consider the optimization problem for static ETTs in Eq. \eqref{eq:general-optimization}, we could start with all ETTs set to 0, and then gradually increase the ETTs until the property is violated. Conversely, if we consider the optimization problem in Eq \eqref{eq:general-optimization} for $\rhoett{}$, we would start with large $\epsilonparam{y_i}{\inequalityprop{}}$ parameters and gradually decrease them until the property is violated. This may be a suitable approach if the system is complex and hard to formally analyze, but may have significant computational requirements. Alternatively, a suitable ETT can be determined by formally analyzing the system and controller dynamics. 
In this paper, we primarily consider the former empirical approach.
In the next section, we show how to obtain a lower bound on the parameters such that when the uncertainty of the measured signals is taken into account then the ETTs will be reduced arbitrarily close to 0 as the system becomes unsafe. 
Later, in Section \ref{sec:case-study} we provide insights on how parameter tuning of both constant ETTs and the proposed ETT regulation method can be performed. We also explore what effects suboptimal parameter configurations can have on the satisfaction of the property and the number of triggered events.

\subsection{Guarantees for inequality property satisfaction detection}\label{sec:parameter-tuning}

Recall that the robustness used to calculate the ETT in Definition \ref{def:ett-regulation-safety-1} is an estimated robustness rather than the actual true robustness. As a result, the property in Proposition \eqref{prop:ett} does not guarantee that the ETTs approach zero when the robustness calculated on the true system state, becomes zero. Because, while the robustness may attain a positive value that is close to zero, because of its associated uncertainty, the real robustness may already have become negative.
According to Eq. \eqref{eq:signal-ett-interval} and as Example \ref{ex:robustness-interval} shows, we allow measurable signals to reside in known intervals. Similarly, the state estimates will also reside in an interval and consequently, the robustness will also reside in an interval which can be calculated using interval arithmetic. This robustness interval is then guaranteed to contain the true robustness. If the robustness interval has a negative lower bound, then we can conclude that the system is potentially unsafe allowing the controller (and the ETT regulation mechanism) to react appropriately. 
\begin{example}[Robustness interval]\label{ex:robustness-interval}
    Consider the states $\propstates{}{} = \begin{bmatrix}
        \hat{x}_1 \\ \hat{x}_2
    \end{bmatrix}$ and the property $\inequalityprop{} = (2 x_1 + 4 x_2) > 9$ and the corresponding robustness $\rho(\inequalityprop{}, \propstates{}{t_k}) = 2 x_1 + 4 x_2 - 9$.
    At time $t_k$, we have the following estimated values of the signals $\hat{x}_{1, t_k} = 3, \hat{x}_{2, t_k} = 1$ corresponding to a robustness of $\rho(\varphi) = 2 \cdot 3 + 4 \cdot 1 - 9 = 1$. We can measure both $x_1$ and $x_2$ and use the ETT regulation mechanism from Definition \ref{def:ett-regulation-safety-1} with the parameters $\epsilon_{x_1} = 3, \epsilon_{x_2} = 3$ providing us with ETTs $\delta_{x_1} = \frac{1}{3}, \delta_{x_2} = \frac{1}{3}$. We assume that the signals $x_1$ and $x_2$ do not change over the next time step and we apply the calculated ETT. Thus the robustness interval at time $t_{k+1}$ becomes $2[3 - \frac{1}{3}, 3 + \frac{1}{3}] + 4[1 - \frac{1}{3}, 1 + \frac{1}{3}] = [10 - \frac{6}{3}, 10 + \frac{6}{3}] = [8, 12]$. As a result, the system could potentially violate $\varphi$ without the remote state estimator or controller being aware which is undesirable.
\end{example}
In the following section, we provide a method for determining the set of minimum $\epsilonparam{y_i}{\inequalityprop{}}$ parameters that ensure that the minimum ETT is always applied before the robustness calculated based on the true system state $\truerobustness{\inequalityprop{}}{t_k}$ reaches zero. 

Let $\rho(\inequalityprop{},\Delta\propstates{}{t_k})$ denote the robustness interval,
where $\Delta\propstates{}{t_k}$ is the state estimate interval vector at time $t_k$. We assume that $\Delta \propstates{}{t_k}$ only contains directly measurable states and that the innovation update error (\cref{eq:innovation}) is used. Thus $\Delta \propstates{}{t_k}$ can be constructed as follows: 
\begin{equation}\label{eq:state-interval}
    \Delta\propstates{}{t_k} = 
    \begin{bmatrix}
        [\propstate{1}{t_k} - \ettnoargs{x_1}{\inequalityprop{}}{t_k}, \propstate{1}{t_k} + \ettnoargs{x_1}{\inequalityprop{}}{t_k}] \\ 
        \vdots \\
        [\propstate{n}{t_k} - \ettnoargs{x_n}{\inequalityprop{}}{t_k}, \propstate{n}{t_k} + \ettnoargs{x_n}{\inequalityprop{}}{t_k}]
    \end{bmatrix}.
\end{equation}

We also assume that it is possible to predict the bounds of the state estimate and consequently the robustness at the next time step $\Delta \boldsymbol{\hat{x}}_{t_{k+1}|t_k}$, given the state estimate information at time step $t_k$:
\begin{equation}\label{eq:predicted-wc-robustness}
    [\nextsteprobustnesswc{\inequalityprop{}}{t_{k+1}|t_k}, \nextsteprobustnessbc{\inequalityprop{}}{t_{k+1}|t_k}] = \rho(\inequalityprop{},\Delta\propstates{}{t_{k+1} | t_k}).
\end{equation}
This prediction can be obtained using set-membership filtering \cite{milaneseOptimalEstimationTheory1991}. Alternatively, one may use a probabilistic state estimation approach such as the KF and use the associated uncertainty related to the state estimate, to construct a confidence interval on the state interval.
In Theorem \ref{thm:parameter-selection}, we present a method for automatically determining the $\epsilonparam{y_i}{\inequalityprop{}}$ parameters from Definition \ref{def:ett-regulation-safety-1} such that the sign of the robustness will always equal the sign of the true robustness, i.e. we guarantee that we detect the satisfaction/violation of a property.

\begin{theorem}\label{thm:parameter-selection}
    Consider an inequality property $\inequalityprop{}$ with the robustness interval width $w(\robustnessintervals{\inequalityprop{}}{t_{k+1}})$ and the robustness of $\inequalityprop{}$ calculated on the true system state $\truerobustness{\inequalityprop{}}{t_{k+1}}$ at time $t_{k+1}$. By setting $\epsilonparam{y_i}{\inequalityprop{}} = \frac{\partial w(\robustnessintervals{\inequalityprop{}}{})}{\partial \ettnoargs{y_i}{\inequalityprop{}}{t_{k+1}}} \lambdaparam{y_i}{\inequalityprop{}}\epsilonparam{\rho}{\inequalityprop{}}$ where $\sum_{y_i \in \propsignals{\inequalityprop{}}} \frac{1}{\lambdaparam{y_i}{\inequalityprop{}}} = 1, \epsilonparam{\rho}{\inequalityprop{}} \geq 1$ and using $\nextsteprobustnesswc{\inequalityprop{}}{t_{k+1}|t_k}$ instead of $\robustness{\inequalityprop{}}{t_k}$ in Definition \ref{def:ett-regulation-safety-1}, we achieve the property
    \begin{equation}\label{eq:robustness-interval-requirements}
        (\robustness{\inequalityprop{}}{t_k} > 0) \Leftrightarrow (\truerobustness{\inequalityprop{}}{t_k} > 0)
    \end{equation} 
    when $\forall x_i \in \propstates{}{} \indent \frac{\partial p(\propstates{}{})}{\partial x_i} = \alpha_i$, the innovation update error is used and $\forall x_i \in \propstates{}{} : \frac{\partial p(\propstates{}{})}{x_i} \neq 0 \text{ } \exists y_i \in \propsignals{\inequalityprop{}} : y_i \equiv x_i$. The latter formalizes that all relevant states are directly measurable. 
\end{theorem}
\begin{proof}
    To satisfy the condition in \cref{eq:robustness-interval-requirements} we predict the robustness interval at time $t_{k+1}$ to determine the necessary ETTs and thus signal intervals at time $t_{k+1}$. 
    As the true robustness $\truerobustness{\inequalityprop{}}{t_{k+1}} \in \robustnessintervals{\inequalityprop{}}{t_{k+1} | t_k}$, we consider the worst-case scenario $(\truerobustness{\inequalityprop{}}{t_{k+1}} = \nextsteprobustnesswc{\inequalityprop{}}{t_{k+1}|t_k})$. Next, we demonstrate that by setting $\epsilonparam{y_i}{\inequalityprop{}} = \frac{\partial w(\robustnessintervals{\inequalityprop{}}{})}{\partial \ettnoargs{y_i}{\inequalityprop{}}{t_{k+1}}} \lambdaparam{y_i}{\inequalityprop{}}\epsilonparam{\rho}{\inequalityprop{}}$ we have:
    \begin{equation}\label{eq:robustness-interval-tuning}
        w(\robustnessintervals{\inequalityprop{}}{t_{k+1}}) \leq \max(\frac{\nextsteprobustnesswc{\inequalityprop{}}{t_{k+1}|t_k}}{\epsilonparam{\rho}{\inequalityprop{}}}, 0).
    \end{equation}
    As the robustness is a linear combination of the bounded scaled state estimates and we use the innovation update error, the upper bound of the robustness interval width (the case when no events are triggered) can be calculated using the state estimate intervals $\Delta \propstates{}{t_{k+1}}$ from \cref{eq:state-interval}. Using the result of Lemma \ref{lem:ett-interval-width}, we have $w(\robustnessintervals{\inequalityprop{}}{t_{k+1}}) = 2|\alpha_1| \ettnoargs{y_1}{\inequalityprop{}}{t_{k+1}} + ... + 2 |\alpha_n| \ettnoargs{y_n}{\inequalityprop{}}{t_{k+1}}$ and thus also $\frac{\partial w(\robustnessintervals{\inequalityprop{}}{})}{\partial \ettnoargs{y_i}{\inequalityprop{}}{t_{k+1}}} = 2 |\alpha_i|$ where $\alpha_i = \frac{\partial \robustness{\varphi}{t_{k+1}}}{y_i}$.
    If we insert $\epsilon_{y_{i}, \inequalityprop{}} = \frac{\partial w(\robustnessintervals{\inequalityprop{}}{})}{\partial \ettnoargs{y_i}{\inequalityprop{}}{t_{k+1}}} \lambdaparam{y_i}{\inequalityprop{}} \epsilonparam{\rho}{\inequalityprop{}}$ in Definition \ref{def:ett-regulation-safety-1} and use $\nextsteprobustnesswc{\inequalityprop{}}{t_{k+1}|t_k}$ instead of $\robustness{\varphi}{t_k}$ we have:
    \begin{align*}
        &w(\robustnessintervals{\inequalityprop{}}{t_{k+1}}) \\
        & \indent \indent = \sum_{y_i \in \propsignals{\inequalityprop{}}} 2|\alpha_i| \frac{\max(\nextsteprobustnesswc{\inequalityprop{}}{t_{k+1}|t_k}, 0)}{2|\alpha_i| \lambdaparam{y_i}{\inequalityprop{}} \epsilonparam{\rho}{\inequalityprop{}}}, \\
        & \indent \indent = \frac{1}{\epsilonparam{\rho}{\inequalityprop{}}} \sum_{y_i \in \propsignals{\inequalityprop{}}} \frac{\max(\nextsteprobustnesswc{\inequalityprop{}}{t_{k+1}|t_k}, 0)}{\lambdaparam{y_i}{\inequalityprop{}}}  \\
        & \indent \indent = \frac{\max(\nextsteprobustnesswc{\inequalityprop{}}{t_{k+1}|t_k}, 0)}{\epsilonparam{\rho}{\inequalityprop{}}}.
    \end{align*}
    Next, assume $\robustness{\inequalityprop{}}{t_{k+1}} \leq 0$. Since $\robustness{\inequalityprop{}}{t_{k+1}} \in \robustnessintervals{\inequalityprop{}}{t_{k+1}}$ and $\robustnessintervals{\inequalityprop{}}{t_{k+1}} \subseteq \robustnessintervals{\inequalityprop{}}{t_{k+1} | t_k}$, and therefore $\nextsteprobustnesswc{\inequalityprop{}}{t_{k+1}|t_k} \leq \robustness{\inequalityprop{}}{t_{k+1}} \leq 0$. Consequently $$(w(\robustnessintervals{\inequalityprop{}}{t_{k+1}}) = 0) \Rightarrow (\robustness{\inequalityprop{}}{t_{k+1}} = \truerobustness{\inequalityprop{}}{t_{k+1}})$$ and thus $(\robustness{\inequalityprop{}}{t_{k+1}} \leq 0) \Rightarrow (\truerobustness{\inequalityprop{}}{t_{k+1}} \leq 0)$. Following the same line of arguments, we can obtain $(\truerobustness{\inequalityprop{}}{t_{k+1}} \leq 0) \Rightarrow (\robustness{\inequalityprop{}}{t_{k+1}} \leq 0)$. 
    If $(\robustness{\inequalityprop{}}{t_k} \leq 0) \Leftrightarrow (\truerobustness{\inequalityprop{}}{t_k} \leq 0)$, then $(\robustness{\inequalityprop{}}{t_k} > 0) \Leftrightarrow ((\truerobustness{\inequalityprop{}}{t_k}) > 0)$ which concludes the proof.
\end{proof}

Do note that the ETT regulation method in Theorem \ref{thm:parameter-selection} does not solve the optimization problem in Eq. \eqref{eq:general-optimization}, but rather it provides a method to determine a sufficiently small ETT such that the system is always aware of whether or not it is potentially violating a property and can act accordingly. The $\epsilonparam{\rho}{\inequalityprop{}}$ parameter in Theorem \ref{thm:parameter-selection}, determines how early the robustness interval starts converging. Larger values of $\epsilonparam{\rho}{\inequalityprop{}}$ correspond to generally lower ETTs and vice versa. We can then adapt the optimization problem in \ref{eq:general-optimization} by setting $\Xi$ to the set of defined $\lambdaparam{y_i}{\inequalityprop{}}$ and $\epsilonparam{\rho}{\inequalityprop{}}$ parameters and adding the constraints in Theorem \ref{thm:parameter-selection}.
We provide a default way to set the $\lambdaparam{y_i}{\inequalityprop{}}$ weights from Theorem \ref{thm:parameter-selection}, in Proposition \ref{prop:robustness-interval-contribution-weight}.

\begin{proposition}\label{prop:robustness-interval-contribution-weight}
    Consider the $\lambdaparam{y_i}{\inequalityprop{}}$ parameters in Theorem \ref{thm:parameter-selection} where we require $\sum_{y_i \in \boldsymbol{y_{\sim\hat{x}_\varphi}}} \frac{1}{\lambdaparam{y_i}{\inequalityprop{}}} = 1$, and $\dim(\lbrace x_i \in \propstates{}{} : \frac{\partial p(\propstates{}{})}{x_i} \neq 0) \rbrace = \dim(\propsignals{\inequalityprop{}})$ (a consequence of the relevant states being directly measurable) where $\dim(\boldsymbol{F})$ is the number of elements in the set $\boldsymbol{F}$. If we set $\lambdaparam{y_i}{\inequalityprop{}} = \dim(\propsignals{\inequalityprop{}}) \text{ } \forall y_i \in \propsignals{\varphi}$ then $\sum_{y_i \in \propsignals{\inequalityprop{}}}\frac{1}{\dim(\propsignals{\inequalityprop{}})} = 1$.
\end{proposition}
Next, we briefly discuss a few remarks regarding some of the assumptions related to Theorem \ref{thm:parameter-selection} followed by how to handle the case when multiple inequality properties depend on one or more of the same signals.

\subsubsection{SOD update error and non-linearities}

Instead of the innovation-based update error, the SOD update error could also be used to bound the state used to calculate the robustness by using the interval determined by $\Delta\hat{x}_{i,t_k} = [y_{i, \tau^i_m} - \delta_{y_i, \inequalityprop{}, t_k}, y_{i, \tau^i_m} + \delta_{y_i, \inequalityprop{}, t_k}]$ where $y_{i, \tau^i_m}$ is the value of $y_i$ at time the previous event-triggering time $\tau^i_m$ for the $i^{th}$ sensor.

Additionally, if the function $p(\boldsymbol{\hat{x}})$ contains non-linear combinations of any elements of $\boldsymbol{\hat{x}}$, the guarantees that Theorem \ref{thm:parameter-selection} provide become more difficult to achieve as we need to consider both the values that can be measured (bounded by Eq. \eqref{eq:predicted-wc-robustness}), as well as the state estimate intervals in the case that events are not triggered. However, we still note that it is likely beneficial to consider the partial derivative of the ETT w.r.t to the robustness interval width when choosing the $\epsilon_{y_i, \inequalityprop{}}$ parameters in such cases as this gives a measure of how ``important'' a given signal is. 

\subsubsection{Multiple properties with overlapping signal sets} 
We now consider the case where multiple inequality PL properties are present which depend on overlapping signal sets. 

We introduce the notation $\ettfinal{y_i}{t_k}$ to denote the final applied ETT to the signal $y_i$ at time $t_k$ based on all system properties $\allprops$. We rely on the assumption that if 
\begin{equation*}
    \begin{split}
        & \forall y_i \in \propsignals{\varphi_b}, \forall t_k \in [t_0, t_{end}] \text{ } \ettfinal{y_i}{t_k} = \ettnoargs{y_i}{\varphi_{b}}{t_k} \\
        & \indent \indent \rightarrow \forall t_k \in [t_0, t_{end}] \robustness{\varphi_b}{t_k} > 0,
    \end{split}
\end{equation*}
then
\begin{equation}\label{eq:lower-ett-still-implies-satisfaction}
    \begin{split}
        & \forall y_i \in \propsignals{\varphi_b}, \forall t_k \in [t_0, t_{end}], \\
        & \indent \forall \ettfinal{y_i}{t_k} \in \lbrace \ettnoargs{y_i}{\varphi_{a}}{t_k} : \ettnoargs{y_i}{\varphi_a}{t_k} \leq \ettnoargs{y_i}{\varphi_{b}}{t_k} \rbrace \\
        & \indent \indent \rightarrow \forall t_k \in [t_0, t_{end}] \robustness{\varphi_b}{t_k} > 0,
    \end{split}
\end{equation}
meaning that an $\ettnoargs{y_i}{\varphi_a}{t_k}$ for a signal $y_i$ defined based on the property $\varphi_{a}$, resulting in a smaller ETT than $\ettnoargs{y_i}{\varphi_{b}}{t_k}$ defined based on another property $\varphi_{b}$, when applied to signal $y_i$, still guarantees the satisfaction of the other property $\varphi_{b}$.
\begin{definition}\label{def:ett-regulation-overalpping-sets}
    For a set of PL properties $\allprops$ with either constant ETTs or ETTs calculated using Definition \ref{def:ett-regulation-safety-1} with overlapping state estimate dependencies, the final ETT of a signal $y_i \in \boldsymbol{y}$ is given by:
    \begin{align}\label{eq:overlapping-sets}
        &  \ettfinal{y_i}{t_k} = \min_{\varphi_j \in \varPhi} \ettnoargs{y_i}{\varphi_j}{t_k} \text{ s.t. } y_i \in \propsignals{\varphi_j}.
    \end{align}
    Eq. \eqref{eq:overlapping-sets} is to be applied after any of the previous or following ETT regulation mechanisms have been applied if multiple separate properties are monitored.
\end{definition}
\begin{lemma}\label{lem:min-ett}
    Definition \ref{def:ett-regulation-overalpping-sets} preserves Theorem \ref{thm:parameter-selection}, if Theorem \ref{thm:parameter-selection} is used to regulate the ETT of the underlying inequality properties and the inequality properties satisfy the conditions in Theorem \ref{thm:parameter-selection}.
\end{lemma}
\begin{proof}
    Using the result of Lemma \ref{lem:ett-interval-width}, we know that $\frac{\partial w(\robustnessintervals{\varphi}{})}{\partial \ettnoargs{y_i}{\varphi}{t_k}} = 2 |\alpha_i|$, meaning that if we decrease the ETT, then $w(\robustnessintervals{\varphi}{t_k})$ decreases aswell which preserves the inequality.  
\end{proof}

\subsection{ETT regulation for PL properties with inequality sub-properties}\label{sec:ett-regulation-propositional}

We now leverage the satisfaction relation of the remaining PL operators to define an ETT regulation mechanism that refines the ETT determined by the underlying inequality property, based on the propositional operators. For the remainder of the paper, we use the shorthand notation $\rhoett{\cdot}$ to denote the application of the robustness-relative ETT regulation policy based on Definition \ref{def:ett-regulation-safety-1} and $\rhoettwc{\cdot}$ to denote the ETT regulation policy based on Theorem \ref{thm:parameter-selection}. In this Section, we consider the $\rhoett{\cdot}$ policy, and later in Section \ref{sec:arb-prop-detection-guarantees} we extend the $\rhoettwc{\cdot}$ from Theorem \ref{thm:parameter-selection} to arbitrary PL properties.
For simplicity and to facilitate intuition, the sub-properties of all propositional properties are initially restricted to inequality properties. Later in this section, we provide a way to compute the ETTs for an arbitrary PL property.

\subsubsection{The $\land$ operator}

By Definition \ref{def:stl-robustness}, for a property $\varphi_{\land} = \varphi_1 \land \varphi_2$ we desire to find an ETT policy $\rhoett{\varphi_{\land, ...}}$ that ensures
\begin{equation}\label{eq:and-ett-regulation}
    \begin{split}
        & \rhoett{\varphi_\land, ...} \Rightarrow (\robustness{\varphi}{t_k} > 0) \\
        & \equiv (\rhoett{\varphi_1, ...} \Rightarrow (\robustness{\varphi_1}{t_k})) \land \\
        & \indent \text{ } (\rhoett{\varphi_2, ...} \Rightarrow (\robustness{\varphi_2}{t_k})),\\
        & \forall t_k \in [t_0, t_{end}].
    \end{split}
\end{equation}

As a result, we can treat the analysis of the sub-properties $\varphi_1$ and $\varphi_2$ to determine suitable static ETTs or ETT regulation parameters, separately. 

\subsubsection{The $\lor$ operator}

According to Definition \ref{def:ett-regulation-safety-1}, for the property $\varphi_{\lor} = \varphi_1 \lor \varphi_2$, we desire to find an ETT policy $\rhoett{\varphi_{\lor}, ...}$ such that
\begin{equation}\label{eq:or-ett-regulation}
    \begin{split}
        & \rhoett{\varphi_\lor, ...} \Rightarrow (\rho(\boldsymbol{\hat{x}_{\varphi_\lor, t_k}}, \varphi_\lor) > 0 ) \\
        & \equiv \rhoett{\varphi_\lor, ...} \Rightarrow \\
        & \indent [(\rho(\boldsymbol{\hat{x}_{\varphi_1, t_k}}, \varphi_1) > 0) \lor (\rho(\boldsymbol{\hat{x}_{\varphi_2, t_k}}, \varphi_2) > 0)], \\
        & \forall t_k \in [t_0, t_{end}].
    \end{split}
\end{equation}

This is a notable difference from the $\land$ operator. We could treat the $\varphi_1$ and $\varphi_2$ sub-properties separately as we did with the $\land$ operator since $(\varphi_1 \land \varphi_2) \Rightarrow (\varphi_1 \lor \varphi_2)$ but we then loose out on potential communication savings as only one property has to be satisfied at any given time. According to Eq. \eqref{eq:or-ett-regulation}, the robustness of one sub-property can become negative, while the overall property is still satisfied. Thus we need to define how the ETT for a sub-property with negative robustness is calculated. We argue that when an inequality property has negative robustness, the ETT should be equal to 0 providing the system with the best possible accuracy of states related to that specific property. Thus we refine Definition \ref{def:ett-regulation-safety-1} by adding:
\begin{equation}\label{eq:min-max}
    \ettpositivenoargs{y_i}{\inequalityprop{}}{t_{k+1}} = \max(\ettnoargs{y_i}{\inequalityprop{}}{t_{k+1}}, 0).
\end{equation}

We wish to leverage the relaxation in Eq. \eqref{eq:or-ett-regulation} to enlarge the ETTs of the inequality sub-property that requires the best accuracy (and thus the most measurements) to satisfy. 
We intuitively want to avoid directly mixing ETTs as different state estimates can have different ranges of values and thus the properties have different ranges of robustness and ETTs. Later in Example \ref{ex:robustness-interval}, we provide some intuition for this decision. To this end, we propose to use the normalized robustness $\zeta_\rho$ which is defined in Definition \ref{def:relative-robustness-interval}, to represent the ``criticality" of an PL inequality property which is then later applied in Definition \ref{def:prop-operators-ett-regulation} to enlarge the ETT of the less critical inequality property. 
\begin{definition}[Normalized robustness $\zeta_\rho$]\label{def:relative-robustness-interval}
    We define the normalized robustness for an inequality property $\inequalityprop{}$ as:
    \begin{align*}
        \relativerobustness{\inequalityprop{}}{t_k} = \frac{\max(\rho(\inequalityprop{}, \boldsymbol{\hat{x}_{t_k}}), 0)}{\rho_{\max}(\inequalityprop{})},
    \end{align*}
    where $\rho_{max}(\inequalityprop{}) > 0$ and is the maximum possible robustness of the inequality property $\inequalityprop{}$.  
\end{definition}

We now define the ETT refinement mechanism for propositional PL properties in Definition \ref{def:prop-operators-ett-regulation}. 
\begin{definition}\label{def:prop-operators-ett-regulation}
    We define the ETT regulation mechanism for propositional property $\varphi$ with inequality sub-properties as a refinement $\ettrefined{y_i}{\varphi}{t_{k+1}}{t_k}{\beta_{\varphi, t_k}}$ of $\ettpositivenoargs{y_i}{\varphi}{t_{k+1}}$ from \cref{eq:min-max} for a signal $y_i$:
    \begin{align*}
        & \ettrefined{y_i}{\varphi}{t_{k+1}}{t_k}{\beta_{\varphi,t_k}} = \\
        & \indent \begin{cases}
            \underset{\varphi' \in \lbrace \inequalityprop{1}, \inequalityprop{2} \rbrace}{\min}(\ettin{y_i}{\varphi'}{t_{k+1}}{t_k}{\beta_{\varphi, t_k}}), \\ 
            \indent \indent \indent \indent \text{if } \varphi = \inequalityprop{1} \lor \inequalityprop{2}, \\
            \underset{\varphi' \in \lbrace \inequalityprop{1}, \inequalityprop{2} \rbrace}{\min}(\ettin{y_i}{\varphi'}{t_{k+1}}{t_k}{0}) \\ 
            \indent \indent \indent \indent \text{if } \varphi = \inequalityprop{1} \land \inequalityprop{2}, \\
        \end{cases} \\
        & \indent \text{where } \beta_{\varphi, t_k} = \underset{\varphi' \in \lbrace \inequalityprop{1}, \inequalityprop{2} \rbrace}{\max}(\relativerobustness{\varphi'}{t_k}), \\
        & \indent \ettin{y_i}{\inequalityprop{}}{t_{k+1}}{t_k}{\beta_{\varphi, t_k}} = \\
        & \indent \begin{cases}
            \ettpositivenoargs{y_i}{\inequalityprop{}}{t_{k+1}} \\
            \indent + \max(\beta_{\varphi, t_k} - \relativerobustness{\inequalityprop{}}{t_k}, 0)\frac{\rho_{max}(\inequalityprop{})}{\epsilonparam{y_i}{\inequalityprop{}}} \\ 
            \indent \indent \indent \indent \text{if } y_i \in \propsignals{\inequalityprop{}}, \\
            \infty \indent \indent \indent \text{otherwise},
        \end{cases} \\
    \end{align*}
    Definition \ref{def:ett-regulation-overalpping-sets} is applied within Definition \ref{def:prop-operators-ett-regulation} using the $\min$ operator for each sub-property. If more than one property is present, Definition \ref{def:ett-regulation-overalpping-sets} should be applied externally.
\end{definition}

\begin{example}[Normalized robustness]\label{ex:normalized-robustness}
    Consider the property $\varphi_{\zeta_\rho} = (x_1 < 1) \lor (x_2 > 1000)$ where the value of $x_1 \in [0, 1.5]$ and the value of $x_2 \in [-1000, 3000]$. We use the parameters $\epsilon_{x_1, \varphi_{\zeta_\rho}} = \epsilon_{x_2, \varphi_{\zeta_\rho}} = 5$. At time $t_k$, $x_{1, t_k} = 1.2$ and $x_{2, t_k} = 1500$, resulting in robustnesses of $\rho(x_1 < 1) = 1 - 1.2 = -0.2$ and $\rho(x_2 > 1000) = 1500 - 1000 = 500$ and ETTs: $\ettnoargs{x_1}{\varphi_{\zeta_\rho}}{t_{k+1}} = \frac{0}{5} = 0$ and $\ettnoargs{x_2}{\varphi_{\zeta_\rho}}{t_{k+1}} = \frac{500}{5} = 100$. As $(x_2 > 1000)$ the overall property is satisfied, but we want to utilize this fact to relax $\ettnoargs{x_1}{\varphi_{\zeta_\rho}}{t_{k+1}}$ as setting $\ettnoargs{x_1}{\varphi_{\zeta_\rho}}{t_{k+1}} = 0$ will likely result in many triggered events. Directly setting $\ettnoargs{x_1}{\varphi_{\zeta_\rho}}{t_{k+1}} = \ettnoargs{x_2}{\varphi_{\zeta_\rho}}{t_{k+1}} = 100$ is likely too large and does not preserve the meaningful range of ETTs for $\ettnoargs{x_1}{\varphi_{\zeta_\rho}}{t_{k+1}}$. Using Definitions \ref{def:relative-robustness-interval} and \ref{def:prop-operators-ett-regulation}, we instead obtain: $\zeta_\rho(x_1 < 1, -1.2) = 0$, $\zeta_\rho(x_2 > 1000, 1500) = \frac{500}{2000}$. Refining $\ettnoargs{x_1}{\varphi_{\zeta_\rho}}{t_{k+1}}$ gives us: $\ettrefinednoargs{x_1}{\varphi_{\zeta_\rho}}{t_{k+1}} = \ettnoargs{x_1}{\varphi_{\zeta_\rho}}{t_{k+1}} + (\frac{500}{2000} - 0)\frac{1}{5} = 0 + \frac{1}{4}\frac{1}{5} = 0.05$ and $\ettrefinednoargs{x_2}{\varphi_{\zeta_\rho}}{t_{k+1}} = \ettnoargs{x_2}{\varphi_{\zeta_\rho}}{t_{k+1}}$. Thus Definition \ref{def:relative-robustness-interval} gives us a way to translate/relate the ``criticality'' of properties between each other. 
\end{example}

In the following Section, we adapt Definition \ref{def:prop-operators-ett-regulation} to arbitrary PL properties.

\subsection{ETT regulation for arbitrary PL properties}\label{sec:ett-regulation-arb-propositional}

\subsubsection{Negation of PL properties}

To regulate the ETT using Definition \ref{def:ett-regulation-safety-1} for an arbitrary PL property which potentially includes negations of logical operators, we transform the property into its Negation Normal Form (NNF), where the negations are propagated to the underlying inequality properties. The robustness of the negated inequality is then used to regulate the ETT according to Definition \ref{def:ett-regulation-safety-1}.

\begin{example}[Negation Normal Form]
    Consider the PL property $\varphi_{\text{NNF}} = (\varphi_1 \land \neg(\varphi_2 \rightarrow \varphi_3))$ where $\varphi_1, \varphi_2$ and $\varphi_3$ are all propositional PL properties. We propagate the negation as follows
    \begin{align*}
        & \neg(\varphi_2 \rightarrow \varphi_3) \\
        & \equiv \neg \neg \varphi_2 \lor \neg\varphi_3 \\
        & \equiv \varphi_2 \lor \neg\varphi_3,
    \end{align*}
    which gives us the final property $\varphi_{\text{NNF}} \equiv \varphi_1 \land (\varphi_2 \lor \neg \varphi_3)$.
\end{example}

Initially, we need to define how to calculate $\relativerobustness{\varphi}{t_k}$ for an arbitrary PL property.

\begin{definition}\label{def:ett-recursive-upstream}
    The Normalized robustness $\relativerobustnessrefined{\varphi}{t_k}$ for an arbitrary PL property $\varphi$ is defined recursively as:
    \begin{align*}
        & \relativerobustnessrefined{\varphi}{t_k} =  \\
        & \indent \begin{cases}
            \max(\relativerobustnessrefined{\varphi_1}{t_k}, \relativerobustnessrefined{\varphi_2}{t_k}) & \text{ if } \varphi = (\varphi_1 \lor \varphi_2), \\
            \min(\relativerobustnessrefined{\varphi_1}{t_k}, \relativerobustnessrefined{\varphi_2}{t_k}) & \text{ if } \varphi = (\varphi_1 \land \varphi_2), \\
            \relativerobustness{\inequalityprop{}}{t_k} & \text{ otherwise},
        \end{cases}
    \end{align*}
    where $\relativerobustness{\inequalityprop{}}{t_k}$ is calculated according to Definition \ref{def:relative-robustness-interval}.
\end{definition}

Next, we present a property of Definition \ref{def:ett-recursive-upstream} that will be useful in proving Theorem \ref{thm:parameter-selection} for propositional and later arbitrary PL properties.
\begin{lemma}\label{lem:relative-robustness-prop}
    $(\relativerobustnessrefined{\varphi}{t_k} > 0) \Leftrightarrow (\robustness{\varphi}{t_k} > 0)$.
\end{lemma}
\begin{proof}
    We prove each case in Definition \ref{def:ett-recursive-upstream}, below. \\

    \paragraph*{\textbf{otherwise} ($\varphi = \inequalityprop{}$)} It follows from Definition \ref{def:relative-robustness-interval} that $\relativerobustnessrefined{\varphi}{t_k}$ preserves the sign since $\rho_{max} > 0$. \\

    \paragraph*{$\varphi = \varphi_1 \lor \varphi_2$} Assume $\relativerobustnessrefined{\varphi}{t_k} > 0$. Then by Definition \ref{def:ett-recursive-upstream} $(\relativerobustnessrefined{\varphi_1}{t_k} > 0) \lor (\relativerobustnessrefined{\varphi_2}{t_k} > 0)$. Consequently, we have $(\robustness{\varphi_1}{t_k} > 0) \lor (\robustness{\varphi_2}{t_k} > 0)$, which through Definition \ref{def:stl-robustness} results in $\robustness{\varphi}{t_k}$. Thus $(\relativerobustnessrefined{\varphi}{t_k} > 0) \Rightarrow (\robustness{\varphi}{t_k} > 0)$. Following the reverse procedure, we can obtain $\robustness{\varphi}{t_k} \Rightarrow \relativerobustnessrefined{\varphi}{t_k}$. \\

    \paragraph*{$\varphi = \varphi_1 \land \varphi_2$} The proof for this case can be obtained following the same procedure as for the $\lor$ operator.

    Thus for any arbitrary propositional property $\varphi$, $(\relativerobustnessrefined{\varphi}{t_k} > 0) \Leftrightarrow (\robustness{\varphi}{t_k} > 0)$.
\end{proof}

After determining $\relativerobustnessrefined{\varphi}{t_k}$ we need a slight adaptation of Definition \ref{def:prop-operators-ett-regulation} to account for nested propositional operators which we provide in Definition \ref{def:arbitrary-prop-operators-ett-regulation}.
\begin{definition}[Arbitrary propositional properties ETT refinement]
    \label{def:arbitrary-prop-operators-ett-regulation}
    Consider the arbitrary propositional property $\varphi$. We define the ETT refinement $\ettrefinedarb{y_i}{\varphi}{t_{k+1}}{t_k}{\beta_{\varphi, t_k}}$ for a signal $y_i \in \propsignals{\varphi}$ at time $t_{k+1}$:
    \begin{align*}
        & \ettrefinedarb{y_i}{\varphi}{t_{k+1}}{t_k}{\beta_{\varphi, t_k}} = \\ 
        & \indent \begin{cases}
            \ettin{y_i}{\varphi}{t_{k+1}}{t_k}{\beta_{\varphi, t_k}} \\
            \indent \indent \indent \text{if } \varphi = \inequalityprop{}, \\
            \underset{\varphi' \in \lbrace \varphi_1, \varphi_2 \rbrace}{\min}(\ettrefinedarb{y_i}{\varphi'}{t_{k+1}}{t_k}{\beta_{\varphi, t_k}}) \\
            \indent \indent \indent \text{if } \varphi = \varphi_1 \land \varphi_2, \\
            \begin{aligned}
                & \underset{(\varphi', \varphi'') \in \lbrace (\varphi_1, \varphi_2), (\varphi_2, \varphi_1) \rbrace}{\min} \\
                & \indent (\ettrefinedarb{y_i}{\varphi'}{t_{k+1}}{t_k}{\max(\beta_{\varphi, t_k}, \relativerobustnessrefined{\varphi''}{t_k})}) \\
            \end{aligned} \\
            \indent \indent \indent \text{if } \varphi = \varphi_1 \lor \varphi_2, \\
        \end{cases}
    \end{align*} 
    where $\ettin{y_i}{\varphi}{t_{k+1}}{t_k}{\beta_{\varphi, t_k}}$ is defined in Definition \ref{def:prop-operators-ett-regulation} and we initialize $\beta_{\varphi, t_k} = \relativerobustnessrefined{\varphi}{t_k}$ of the top-level property $\varphi$.
    The $\min$ operations take care of potentially multiple defined ETTs for a given signal, and thus Definition \ref{def:ett-regulation-overalpping-sets} is implicitly applied.
\end{definition}
An example of the application of Definition \ref{def:ett-recursive-upstream} and the propagation of $\relativerobustnessrefined{\varphi}{t_k}$ in Definition \ref{def:arbitrary-prop-operators-ett-regulation} is visualized in Fig. \ref{fig:property-visualization}.

\begin{figure}[h]
    \centering
    \begin{subfigure}{0.4\textwidth}
        \includegraphics[width=\textwidth]{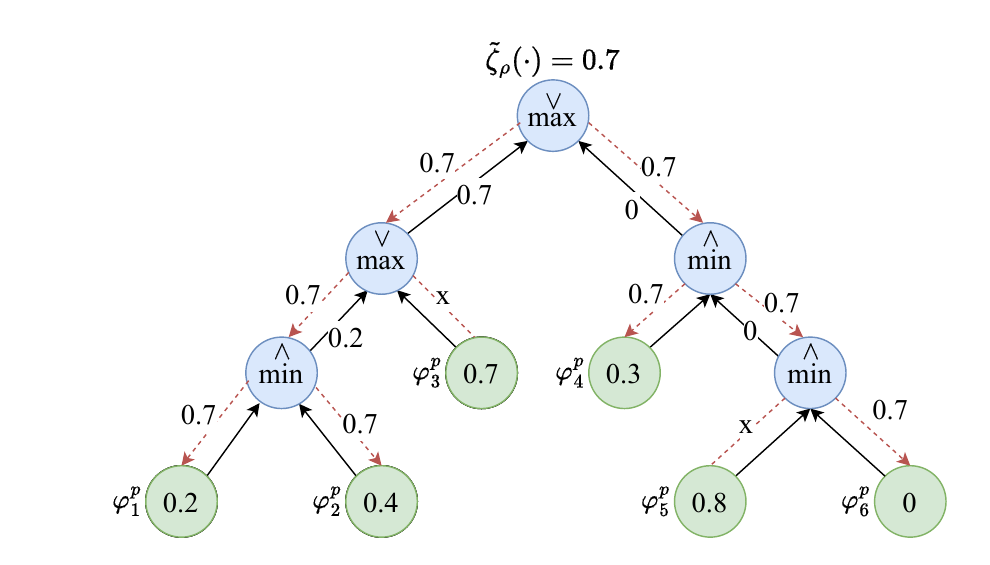}
        \caption{}
        \label{fig:stl-property-binary-tree}
    \end{subfigure}
    \caption{A binary tree representation of the PL property \break$((\varphi_1 \land \varphi_2) \lor \varphi_3) \lor (\varphi_4 \land (\varphi_5 \land \varphi_6))$. The solid black arrows indicate the application of Definition \ref{def:ett-recursive-upstream} and the red dotted arrows indicate the applied $\beta$ value in Definition \ref{def:arbitrary-prop-operators-ett-regulation} for each inequality property. 
    }
    \label{fig:property-visualization}
\end{figure}

Definition \ref*{def:arbitrary-prop-operators-ett-regulation} concludes the proposed ETT regulation mechanism $\rhoett{\cdot}$. In the following Section, we extend the property in Theorem \ref{thm:parameter-selection} arbitrary PL properties using a modified version of $\rhoett{\cdot}$.
\subsection{Providing detection guarantees for arbitrary PL properties}\label{sec:arb-prop-detection-guarantees}

We start by adapting several Definitions from Sections \ref{sec:ett-regulation-propositional} and \ref{sec:ett-regulation-arb-propositional}.
We denote the ETT $\geq 0$ calculated using the result of Theorem \ref{thm:parameter-selection} as:
\begin{equation}\label{eq:thm-1-ett}
    \thmett{y_i}{\inequalityprop{}} = \frac{\max(\nextsteprobustnesswc{\inequalityprop{}}{t_{k+1}|t_k}, 0)}{\frac{\partial w(\robustnessintervals{\inequalityprop{}}{})}{\partial \ettnoargs{y_i}{\inequalityprop{}}{t_{k+1}}} \lambdaparam{y_i}{\inequalityprop{}} \epsilonparam{\rho}{\inequalityprop{}}}.
\end{equation}
Next, we adapt the normalized robustness from Definitions \ref{def:relative-robustness-interval} and \ref{def:ett-recursive-upstream}.

\begin{definition}\label{def:relative-robustness-interval-wc}
    We define the worst-case normalized robustness of an inequality property $\inequalityprop{}$ as:
    \begin{align*}
        \relativerobustnessnextstep{\inequalityprop{}}{t_{k+1}|t_k} = \frac{\max(\nextsteprobustnesswc{\inequalityprop{}}{t_{k+1}|t_k}, 0)}{\rho_{max}(\inequalityprop{})}
    \end{align*}
\end{definition}
\begin{definition}\label{def:relative-robustness-interval-min}
    We define the worst-case normalized robustness $\relativerobustnessrefinednextstep{\varphi}{t_{k+1}|t_k}$ of an arbitrary PL property as and adaptation of $\relativerobustnessrefined{\varphi}{t_{k}}$ as follows:
    \begin{align*}
        & \relativerobustnessrefinednextstep{\varphi}{t_{k+1}|t_k} = \\
        & \indent \begin{cases}
            \underset{\varphi' \in \lbrace \varphi_1, \varphi_2 \rbrace}{\max}\relativerobustnessrefinednextstep{\varphi'}{t_{k+1}|t_k} & \text{ if } \varphi = \varphi_1 \lor \varphi_2, \\
            \underset{\varphi' \in \lbrace \varphi_1, \varphi_2 \rbrace}{\min}\relativerobustnessrefinednextstep{\varphi'}{t_{k+1}|t_k} & \text{ if } \varphi = \varphi_1 \land \varphi_2, \\
            \relativerobustnessnextstep{\varphi}{t_{k+1}|t_k} & \text{ if } \varphi = \inequalityprop{}.
        \end{cases}
    \end{align*}
\end{definition}
Next, we adapt $\ettin{y_i}{\inequalityprop{}}{t_{k+1}}{t_k}{\beta_{\varphi, t_k}}$ from Definition \ref{def:prop-operators-ett-regulation}.
\begin{definition}\label{def:ett-in-wc}
    We define an adapted version of $\ettin{y_i}{\inequalityprop{}{}}{t_{k+1}}{t_k}{\beta_{\varphi, t_k}}$ from Definition \ref{def:prop-operators-ett-regulation} to help provide the guarantees in Theorem \ref{thm:parameter-selection}.
    \begin{align*}
         & \ettinwc{y_i}{\varphi}{t_{k+1}}{t_{k+1}|t_k}{\underline{\beta}_{\varphi, t_{k+1}|t_k}} = \\
         & \indent \begin{cases}
            \thmett{y_i}{\inequalityprop{}} \\
            \indent + \max(\underline{\beta}_{\varphi, t_{k+1}|t_k} - \relativerobustnessnextstep{\inequalityprop{}}{t_{k+1}|t_k}, 0)\frac{\rho_{max}(\inequalityprop{})}{\epsilonparam{y_i}{\inequalityprop{}}} \\ 
            \indent \indent \indent \indent \text{if } y_i \in \propsignals{\inequalityprop{}}, \\
            \infty \indent \indent \indent \text{otherwise}.
         \end{cases}
    \end{align*}
\end{definition}
We are now ready to present $\rhoettwc{\cdot}$ for arbitrary PL properties. Due to space constraints, we do not provide the explicit definition, but instead list the changes necessary to Definition \ref{def:arbitrary-prop-operators-ett-regulation}, in Theorem \ref{thm:prop-thm-1-preservation}.
\begin{theorem}\label{thm:prop-thm-1-preservation}
    Definition \ref{def:arbitrary-prop-operators-ett-regulation} ensures 
    \begin{equation}\label{eq:prop-property-convergence}
        (\robustness{\varphi}{t_k} > 0) \Leftrightarrow (\truerobustness{\varphi}{t_k} > 0)
    \end{equation}  
    for an arbitrary propositional property $\varphi$ if $\relativerobustnessrefinednextstep{\varphi}{t_{k+1}|t_k}$ is used instead of $\relativerobustnessrefined{\varphi}{t_k}$, the inequality properties satisfy the constraints in Theorem \ref{thm:parameter-selection} and $\ettinwc{y_i}{\inequalityprop{}}{t_{k+1}}{t_{k+1}|t_k}{\underline{\beta}_{\varphi, t_{k+1}|t_k}}$ is used instead of $\ettin{y_i}{\varphi}{t_{k+1}}{t_k}{\beta_{\varphi,t_k}}$.
\end{theorem}
\begin{proof}

    We prove that using the ETT regulation described in Theorem \ref{thm:prop-thm-1-preservation} ensures $(\robustness{\varphi}{t_k} \leq 0) \Leftrightarrow (\truerobustness{\varphi}{t_k} \leq 0)$ which then guarantees $(\robustness{\varphi}{t_k} > 0) \Leftrightarrow (\robustness{\varphi}{t_k} > 0)$. 
    Using the same technique as in Lemma \ref{lem:relative-robustness-prop}, we can show that
    \begin{equation}
        (\relativerobustnessrefinednextstep{\varphi}{t_{k+1}|t_k} = 0) \Leftrightarrow (\nextsteprobustnesswc{\varphi}{t_{k+1}|t_k} \leq 0).
    \end{equation}
    Now assume that $\truerobustness{\varphi}{t_{k+1}} \leq 0$. Since $\truerobustness{\varphi}{t_{k+1}} \in \robustnessintervals{\varphi}{t_{k+1}|t_k}$, then $\nextsteprobustnesswc{\varphi}{t_{k+1}|t_k} \leq 0$.
    We now show that the ETTs of the inequality property that produces the resulting robustness remain unaffected by the operations in Definition \ref{def:prop-operators-ett-regulation} in the case of negative robustness. We denote the overall property $\varphi_o$ and the inequality property that produces the robustness value of the overall property, $\inequalityprop{\rho}$.
    We consider the $\lor$ and $\land$ operators in turn.

    \paragraph*{$\varphi = \varphi_1 \lor \varphi_2$}
    In order for $\inequalityprop{\rho}$ to be a sub-property of $\varphi$, then $\robustness{\varphi}{t_{k+1}} = \robustness{\inequalityprop{\rho}}{t_{k+1}}$. Since we assume $\truerobustness{\varphi}{t_{k+1}} \leq 0$, this entails $\robustness{\varphi_1}{t_{k+1}} \leq 0$ and $\robustness{\varphi_2}{t_{k+1}} \leq 0$ which again entails $\relativerobustnessrefinednextstep{\varphi}{t_{k+1}|t_{k}} = 0$. Thus for the sub-property containing or equal to $\inequalityprop{\rho}$, the propagated $\underline{\beta}_{t_{k+1} | t_{k}}$ value is always equal to 0.
    
    \paragraph*{$\varphi = \varphi_1 \land \varphi_2$}

    According to Definition \ref{def:prop-operators-ett-regulation}, $\underline{\beta}_{\varphi, t_{k+1}|t_k}$ is simply propagated and is equal to 0 if the sub-properties contain or are equal to $\inequalityprop{\rho}$.

    Since the $\underline{\beta}_{t_k|t_{k-1}}$ value propagated to $\inequalityprop{\rho}$ will always be 0 in the case of overall negative robustness the term $\max(\underline{\beta}_{t_{k+1}|t_k} - \relativerobustnessrefinednextstep{\inequalityprop{\rho}}{t_{k+1}|t_k}, 0)$ in Definition \ref{def:ett-in-wc} will equal 0 meaning that the ETTs related to $\inequalityprop{\rho}$ are unchanged and thus satisfy Theorem \ref{thm:parameter-selection} for the inequality property $\inequalityprop{\rho}$ and thus for the overall property $\varphi_o$. 

    \end{proof}

This concludes the presentation of our proposed ETT regulation mechanism for PL properties.
In the following Section, we evaluate the TT, CETT, $\rhoett{\cdot}$ and $\rhoettwc{\cdot}$ in a simulated case study where we explore various parameter configurations and their impact on the robustness and number of triggered events.

\section{Case study}\label{sec:case-study}

To evaluate our proposed threshold regulation mechanism, a simulated vehicle convoy in an adaptive cruise control (ACC) scenario is used. Initially, a single lane with two vehicles is considered to evaluate the ETT regulation mechanism in Definition \ref{def:ett-regulation-safety-1} for a single inequality property. To evaluate more complex propositional properties, we extend the scenario to include an additional lane with faster vehicles. Fig. \ref{fig:scenario-overview} shows an overview of the two setups. Next, we describe the simulation and scenario setup.

\begin{figure}[h]
    \centering
    \includegraphics[width=0.5\textwidth]{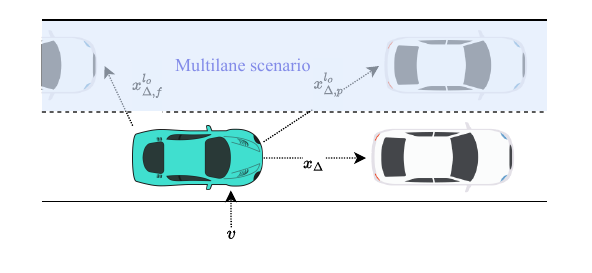}
    \caption{An overview of the case-study simulated scenarios. The blue vehicle corresponds to the ACC vehicle which can measure its own speed ($v$), the distance to the preceding vehicle in the same lane $x_{p}$ and the distance to the preceding $x_{l_o, p}$ and following $x_{l_o, f}$ vehicles in the fast lane in the multi-lane scenario.}
    \label{fig:scenario-overview}
    \vspace*{-10pt}
\end{figure}

\subsection{Simulation setup}

All vehicles in the convoy can be described by the linear discrete-time state-space model specified in Eq. \ref{eq:lti}.

The noise is assumed to be standard Gaussian noise with zero mean and the specified variance.
For the remaining vehicles, we use models without process noise or drift to control their evolution over time. The ACC vehicle tracks itself and all other relevant vehicles using the above model using a separate Kalman filter (KF) for each vehicle. The $A$ and $B$ matrices are the same for all instances and are given by
\begin{equation}
    A = \begin{bmatrix}
        1 & T_s \\ 
        0 & 1 
    \end{bmatrix},
    B = \begin{bmatrix}
        \frac{1}{2} T_s^2 \\
        T_s
    \end{bmatrix},
\end{equation} 
where $T_s$ is the sampling interval. As the ACC vehicle does not know the control input for the other vehicles, $B$ can be disregarded when tracking other vehicles. The $C$ matrices for the ACC vehicle and remaining tracked vehicles denoted with an $_o$ subscript, the process noise of the ACC vehicle $\boldsymbol{w_{ACC}}$ and other vehicles $\boldsymbol{w_{o}}$, and the measurement noise of the speed $\boldsymbol{r_v}$ and distance $\boldsymbol{r_{x_\Delta}}$ sensors are given by 
\begin{align*}
    & C_{ACC} = \begin{bmatrix}
        0 \\
        1
    \end{bmatrix},
    C_o = \begin{bmatrix}
        1 \\ 
        0
    \end{bmatrix},
    \boldsymbol{w_{ACC}} = \begin{bmatrix}
        2.5 \cdot 10^{-9} & 5.0 \cdot 10^{-7}\\
        5.0 \cdot 10^{-7} & 1 \cdot 10^{-4}
    \end{bmatrix}, \\
    & \boldsymbol{w_o} = 10 \cdot \boldsymbol{w_{ACC}},
    \indent \boldsymbol{r_{x_{\Delta}}} = \boldsymbol{r_{v}} = 0.1.
\end{align*} 
To control the ACC vehicle acceleration, we use the Intelligent Driver Model (IDM) \cite{treiberCongestedTrafficStates2000}. The IDM attempts to keep a speed-dependent distance to the preceding vehicle given by $x_{ss}(d_0, v, T) = d_0 + vT$ where $d_0$ is a constant, $v$ is the speed of the ACC vehicle and $T$ is a parameter known as the time headway. We choose $d_0 = 2.7m$ and $T = 2s$.  

The vehicles all start with the same initial speed of $30 m/s$ and are placed at intervals of $x_{ss}(\cdot) + 20m$. The lead vehicle in the ACC vehicle lane has a predefined control schedule where it keeps a constant speed for the first $20$ seconds as the ACC vehicle catches up whereafter it brakes with an acceleration of $a_{min} = -5 m/s^2$ for $5$ seconds and then accelerates to the original speed over $10$ seconds at $a_{max} = 2.5 m/s^2$ which corresponds to the minimum and maximum acceleration of the ACC vehicle respectively. Thus, the ACC vehicle will experience a variety of more and less safety-critical scenarios.

To determine the update error ($e(\cdot)$ in Eq. \eqref{eq:static-ett}), we use the innovation-based update error from Eq. \ref{eq:innovation}. Since the prediction model of the preceding vehicle at the ACC vehicle lacks the control input information, the model is inherently more inaccurate and will lead to significantly more transmitted measurements compared to the sensor measuring the speed of the ACC vehicle. To add some importance to the speed sensor, we simulate an overestimation of the wind resistance leading to greater acceleration than expected. 

For the sake of simplicity, we assume that the measurements, if they satisfy the event-triggering condition, are available instantly at the remote state estimator. Additionally, the information necessary to evaluate the event-triggering condition is assumed to always be available at the smart sensors without the need for additional communication. For single-signal properties, the latter assumption is valid, but it may not be the case for multiple-signal properties. We discuss this further in Section \ref{sec:discussion}.

In Section \ref{sec:inequality-stl}, we compare the TT, CETT, $\rhoett{\cdot}$ and $\rhoettwc{\cdot}$ approaches for an inequality property. Later in Section \ref{sec:temporal-propositional-properties}, we compare the TT, CETT and $\rhoett{\cdot}$ policies for a more complex propositional property involving a multilane scenario.

\subsection{Inequality Property for Single Lane Vehicle Convoy} \label{sec:inequality-stl}

As perfect knowledge of the preceding vehicle is not available, the controller will not be able to keep the desired distance $x_{ss}(\cdot)$ when the lead vehicle brakes. As a result, we construct our inequality property: $\varphi_a = x_{\Delta} > x_{ss}(d_\varphi, v, T)$ where $x_{\Delta} = x_p - x$, $x_p$ and $x$ are the position of the preceding and ACC vehicle, respectively, and $0 \leq d_\varphi < d_0$. The reason for choosing $d_\varphi < d_0$ is that the IDM will not be able to keep the desired distance to the lead vehicle as the ACC vehicle has imperfect knowledge of the speed and acceleration of the lead vehicle. Thus the target distance must be larger than the safe distance. We choose $d_\varphi = 0m$. This gives us the robustness: 
$$\robustness{\varphi}{t_k} = x_{\Delta, t_k} - x_{ss}(0, v_{t_k}, T).$$ 

Initially, we consider the feasibility problem in Eq. \eqref{eq:ett-feasibility}, where we aim to find a sampling frequency that enables the system to satisfy the property. We test two sampling intervals $T_s$ of 0.01s and 0.02 by running 20 simulations of the system for each configuration. We then determine the minimum overall robustness $\rho_{min}$ calculated using the true system state produced for each configuration accross simulations. Setting $T_s = 0.01s$ resulted in $\rho_{min} = \textbf{0.36}$ while setting $T_s = 0.02s$ resulted in $\rho_{min} = \textbf{-0.60}$ (negative robustness), wherefore we choose $T_s = 0.01s$ for all remaining experiments for all ETT regulation policies. 
Note that $T_s$ could potentially be increased to somewhere between $0.01$ and $0.02$, but the number of transmissions cannot be reduced by more than a factor of two given the results of this test.

Next, we consider the CETT, $\rhoett{\cdot}$, and $\rhoettwc{}$ policies for the inequality property $\varphi_a$.
We choose $\propsignals{\varphi_a} = \lbrace v, x_{\Delta}\rbrace$ and test a variety of different combinations of $\epsilonparam{y_i}{\varphi_a}$, $\lambdaparam{y_i}{\varphi_a}$ and $\epsilonparam{\rho}{\varphi_a}$ parameters and constant ETTs using a grid-search technique. We run 20 simulations for each set of parameters and the minimum robustness over all simulations is recorded in Table \ref{tab:results}. Table \ref{tab:results} shows that the ETT policies perform significantly better than the TT policy but also achieve a smaller $\rho_{min}$ for the corresponding configuration. The $\rho$ETT policy triggers approximately \textbf{41.8\%} fewer events compared to the CETT policy while the $\rhoettwc{\cdot}$ policy triggers \textbf{28.4\%} fewer events than the CETT policy.

As the results in Table \ref{tab:results} are found using a brute-force parameter search, they are not guaranteed to be the global minimum. However, the difference in the number of triggered events between the TT, CETT and $\rho ETT$ is large enough to demonstrate a significant advantage of the $\rho$ETT policy over the TT and CETT. Additionally, the ETT policies seem to lead to a smaller $\rho_{min}$ value compared to the TT policy. In a real-world scenario, we may not be interested in the robustness being this close to zero, and we can instead revise the optimization problem in Eq. \eqref{eq:general-optimization} to ensure that $\robustness{\varphi}{t_k} > \eta$ where $\eta > 0$. 

\paragraph*{Parameter exploration}
For the $\rho$ETT regulation mechanism, we provide a contour plot depicting parameter combinations that result in a positive minimum robustness value along with the resulting number of triggered events in Fig. \ref{fig:inequality-property-epsilon-tradeoff}. Fig. \ref{fig:inequality-property-epsilon-tradeoff} shows that generally larger $\epsilonparam{y_i}{\varphi}$ values result in more triggered events regardless of the signal as we would expect. However, increasing the $\epsilonparam{v}{\varphi_a}$ value leads to a smaller increase in the number of triggered events compared to increasing $\epsilonparam{{x_\Delta}}{\varphi_a}$ by the same amount. The figure also shows that when either of the $\epsilonparam{y_i}{\varphi}$ parameters are below a certain threshold, larger values of the other $\epsilonparam{y_i}{\varphi}$ parameter will never result in positive robustness. Additionally, we see that a tradeoff is present in the lower left corner of the feasible set of solutions. Smaller values of either $\epsilonparam{y_i}{\varphi}$ parameter, require larger values for the other parameter to ensure satisfaction. Fig. \ref{fig:inequality-property-epsilon-tradeoff} also shows that the assumption in Eq. \eqref{eq:lower-ett-still-implies-satisfaction} is always true for the $\epsilonparam{x_\Delta}{\varphi_a}$ signal but not necessarily for the $\epsilonparam{v}{\varphi_a}$ parameter. For certain parameter configurations, some larger values of the $\epsilonparam{v}{\varphi_a}$ (resulting in smaller ETTs) do not enable the controller to satisfy the property, while other smaller values of $\epsilonparam{v}{\varphi_a}$ do. This is likely due to the introduced random noise and the relatively small number of simulations.
Due to space constraints, we do not provide the plot for the CETT policy, but the experiments conducted for the CETT policy show a similar tradeoff.

\begin{figure}[h]
    \centering
    \includegraphics[width=0.45\textwidth]{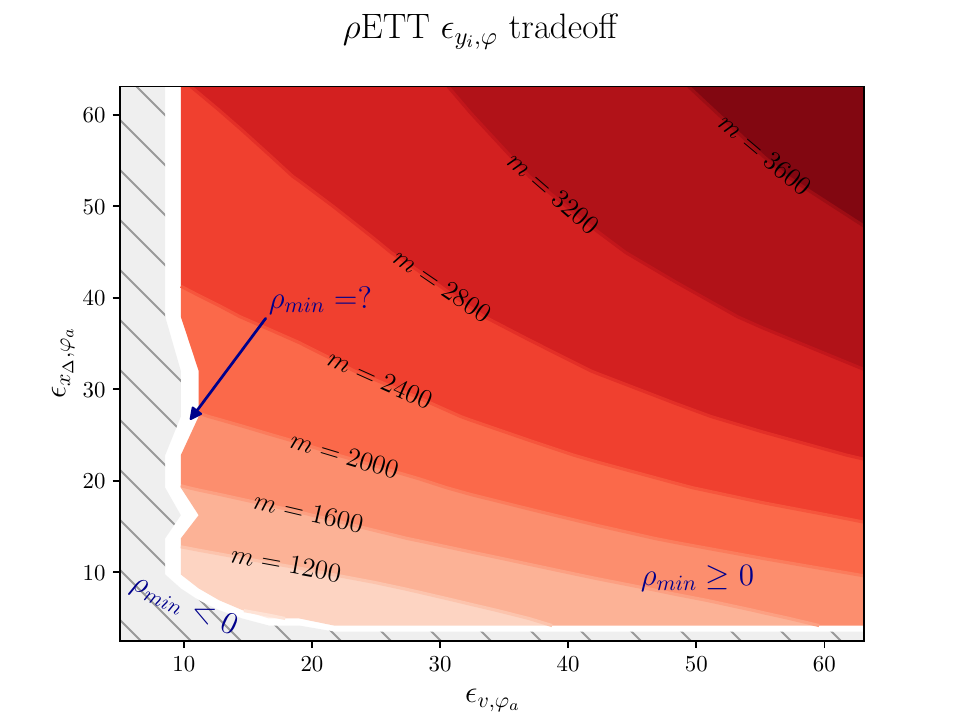}
    \caption{A contour plot of $\epsilonparam{v}{\varphi_a}$ and $\epsilonparam{x_\Delta}{\varphi_a}$ parameter combinations where the red area indicates a feasible configuration and the grey hatched area indicates an infeasible configuration. The white area indicates untested parameter configurations. The shade of red indicates the average number of transmitted measurements.}
    \label{fig:inequality-property-epsilon-tradeoff}
\end{figure}

In Fig. \ref{fig:rho-m-tradeoff}, we provide a plot that explores the tradeoff between the achieved minimum robustness and the number of transmitted measurements for all transmission and ETT regulation strategies. The results shown in Fig. \ref{fig:rho-m-tradeoff} are the results corresponding to the Pareto front (i.e. the set of parameters that achieve a given $\rho_{min}$ value with the fewest number of transmitted measurements). Fig. \ref{fig:rho-m-tradeoff} shows that more transmitted measurements generally lead to larger $\rho_{min}$ values for the Pareto optimal parameter configurations. However, the number of packets increases significantly as we achieve a larger $\rho_{min}$ value resulting in a small performance gain per additional transmitted packet. Generally, the $\rho$ETT achieves the best tradeoff out of all ETT regulation mechanisms while the $\rhoettwc{\cdot}$ policy produces slightly worse results in terms of the number of transmitted packets. This is to be expected since the ETT is regulated based on a more conservative state estimate. The percentage-wise reduction in the number of triggered events is somewhat consistent around \textbf{40\%} for the $\rho$ETT policy compared to CETT which also applies to the best found configuration in Table \ref{tab:results}. 
\begin{figure}[h]
    \centering
    \includegraphics[width=0.48\textwidth]{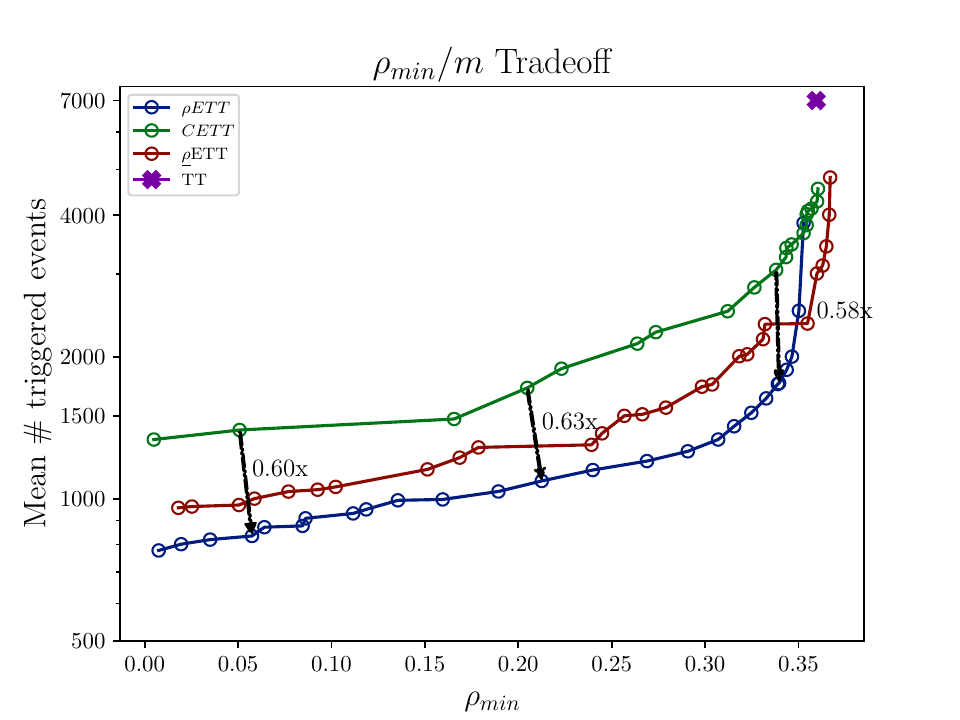}
    \caption{A plot of non-dominated parameter configurations for the CETT and $\rho$ETT ETT regulation strategies, where $m$ is the average number of transmitted packets and $\rho_{min}$ is the minimum robustness evaluated on the true system state over all simulations for a given configuration. Several black arrows indicate what fraction of packets are needed for the $\rho$ETT method to achieve similar minimum robustness values compared to the CETT approach.}
    \label{fig:rho-m-tradeoff}
    \vspace*{-10pt}
\end{figure}
\begin{table}[]
    \centering
    \begin{tabular}{|l|c||lc|}
        \hline
        \multirow{2}{*}{\textbf{Strategy}} & \multirow{2}{*}{\textbf{Parameters}} & \multicolumn{2}{c|}{$\boldsymbol{\varphi_a}$ (Single lane)}  \\ \cline{3-4} 
        & & \multicolumn{1}{c|}{$\rho_{min}$} & $m \pm \sigma(m)$  \\ \hline
        TT & $T_s = 0.01s$ & \multicolumn{1}{l|}{0.36} & 7000 \\ \hline
        CETT & \makecell{$\delta_{v, \varphi_a} = 0.16$ \\ $\delta_{x_\Delta, \varphi_a} = 0.50$} & \multicolumn{1}{l|}{0.005} & $1336 \pm 27$ \\ \hline
        $\rho$ETT & \makecell{$\epsilonparam{v}{\varphi_a} = 16.64$ \\ $\epsilon_{x_{\Delta}} = 4.95$}& 
            \multicolumn{1}{l|}{0.007} &  $777 \pm 14$  \\ \hline
        $\rhoettwc{\cdot}$ & \makecell{$\epsilonparam{v}{\varphi_a} = 13.38$ \\ $\epsilon_{x_{\Delta}} = 3.95$} & \multicolumn{1}{l|}{0.018} & $957 \pm 14$ \\ 

        \hline \hline

        \multirow{2}{*}{} & \multirow{2}{*}{} & \multicolumn{2}{c|}{$\boldsymbol{\varphi_b}$ (Multilane critical)}  \\ \cline{3-4}   
        & & \multicolumn{1}{c|}{$\rho_{min}$} & $m \pm \sigma(m)$  \\ \hline
        TT & $T_s = 0.01s$ & \multicolumn{1}{l|}{3.60}& $14000$ \\ \hline
        CETT & \makecell{$\delta_{v} = 0.16$ \\ $\delta_{x_\Delta} = 0.50$ \\ $\delta_{x^{l_o}_{\Delta,p}} = 0.50$ \\ $\delta_{x^{l_o}_{\Delta,f}} = 2.00 $} & \multicolumn{1}{l|}{3.59} & $1527 \pm 23$ \\ \hline
        $\rho$ETT & \makecell{$\epsilonparam{v}{\varphi_{b_1}} = \epsilonparam{v}{\varphi_{b_3}} = 16.64$ \\ $\epsilonparam{x_\Delta}{\varphi_{b_1}} = 4.95$ \\ $\epsilonparam{x^{l_o}_{\Delta, p}}{\varphi_{b_3}} = 4.95$ \\ $\epsilonparam{x^{l_o}_{\Delta, f}}{\varphi_{b_2}} = 4.03$} & \multicolumn{1}{l|}{3.21} & $166 \pm 7$ \\ \hline
        $\rho$ETT (No $\lor$) & \makecell{$\epsilonparam{v}{\varphi_{b_1}} = \epsilonparam{v}{\varphi_{b_3}} = 16.64$ \\ $\epsilonparam{x_\Delta}{\varphi_{b_1}} = 4.95$ \\ $\epsilonparam{x^{l_o}_{\Delta,p}}{\varphi_{b_3}} = 4.95$ \\ $\epsilonparam{x^{l_o}_{\Delta,f}}{\varphi_{b_2}} = 2.51$} & \multicolumn{1}{l|}{3.41}& $4846 \pm 11$  \\ 
        
        \hline \hline

        \multirow{2}{*}{} & \multirow{2}{*}{} & \multicolumn{2}{c|}{$\boldsymbol{\varphi_b}$ (Multilane non-crit.)}  \\ \cline{3-4}   
        & & \multicolumn{1}{c|}{$\rho_{min}$} & $m \pm \sigma(m)$  \\ \hline
        TT & $-\parallel-$ & \multicolumn{1}{l|}{20.25}& $14000$ \\ \hline
        CETT & $-\parallel-$ & \multicolumn{1}{l|}{20.25} & $1502 \pm 24$ \\ \hline
        $\rho$ETT & $-\parallel-$ & \multicolumn{1}{l|}{20.25} & $55 \pm 2$ \\ \hline
        $\rho$ETT (No $\lor$) & $-\parallel-$ & \multicolumn{1}{l|}{20.25}& $1927 \pm 6$  \\ \hline   
    \end{tabular}
    \caption{An overview of the number of triggered events, minimum robustness and parameters for the found optimal configurations of the TT, CETT and $\rho ETT$ ETT policies. The average number of transmitted measurements and one standard deviation over all simulations for a specific configuration is denoted by $m \pm \sigma(m)$ and $\rho_{min}$ is the minimum robustness over all simulations calculated on the true state. The presence of $-\parallel-$ indicates that the same $\epsilonparam{y_i}{\varphi}$ parameters are used for the non-critical multilane scenario as for the critical scenario.} 
    \label{tab:results} 
    \vspace*{-10pt}
\end{table}

We now further investigate the $\rhoettwc{\cdot}$ policy. Fig. \ref{fig:lambda-tradeoff} depicts results for different values of $\epsilonparam{\rho}{\varphi}$ and $\lambdaparam{y_i}{\varphi}$ combinations. Fig. \ref{fig:lambda-tradeoff} shows that even though the overall uncertainty is fixed relative to the robustness, different $\lambdaparam{y_i}{\varphi}$ combinations produce both different average minimum robustness and a different number of triggered events. We see that for the most part, especially for lower $\epsilon_{\rho, \varphi}$ values, there is a large difference in the number of transmitted packets and some difference in resulting robustness for different $\lambdaparam{y_i}{\varphi}$ parameter combinations. Furthermore, Fig. \ref{fig:lambda-tradeoff} shows the more transmitted packets do not necessarily lead to larger average $\rho_{min}$ values if the parameters are configured sub-optimally. Despite the irregularities in Fig. \ref*{fig:lambda-tradeoff}, there could be some potential to apply a more intelligent numerical optimization algorithm to find locally optimal parameter configurations as there seems to be some structure to the parameter space and the cost in terms of $\rho_{min}$ and the number of triggered events.

\begin{figure}[h]
    \centering
    \includegraphics[width=0.45\textwidth]{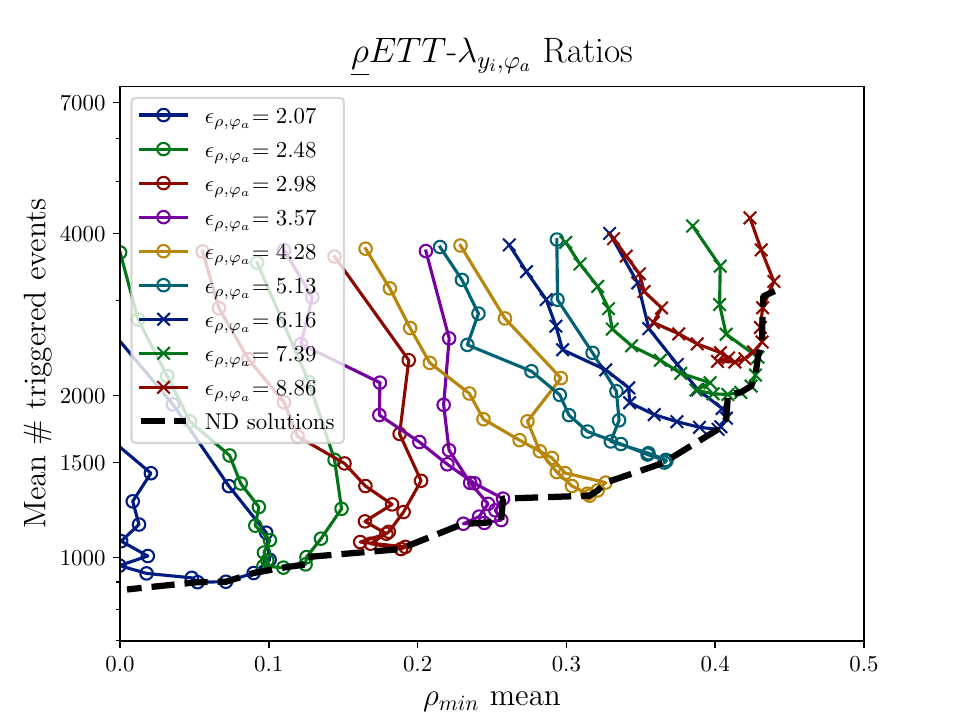}
    \caption{A plot of the resulting mean $\rho_{min}$ and corresponding mean number of transmitted packest $m$ for different values of $\epsilonparam{\rho}{\varphi_a}$ and $\lambdaparam{y_i}{\varphi_a}$ parameters for $\rhoettwc{\cdot}$. The black ND solutions series corresponds to the non-dominated solutions for the average $\rho_{min}$ values as opposed to Fig. \ref{fig:rho-m-tradeoff} where the overall $\rho_{min}$ was used to find the non-dominated solutions.}
    \label{fig:lambda-tradeoff}
    \vspace*{-10pt}
\end{figure}

\subsection{Propositional Property for Multi-lane Vehicle Convoy}\label{sec:temporal-propositional-properties}
To evaluate the ETT regulation mechanisms proposed in Section \ref{sec:ett-regulation-propositional}, we extend the scenario from Section \ref{sec:inequality-stl} to a two-lane ACC scenario. We add a fast lane with additional faster vehicles and enable the ACC vehicle to measure the distance to the two closest vehicles ($x^{l_o}_{\Delta, p}, x^{l_o}_{\Delta, f}$ in preceding and following the ACC vehicle respectively) using two additional sensors. 
We add additional control behavior to allow the ACC vehicle to change to the fast lane. The vehicles in the fast lane travel at a constant speed and are placed sufficiently far apart to allow for safe overtaking when the lead vehicle brakes.

We construct the PL property \break $\varphi_{b} = ((\varphi_{b_1}) \lor (\varphi_{b_2} \land \varphi_{b_3}))$ where $\varphi_{b_1} = x_{\Delta} > x_{ss}(d_{\varphi}, v, T)$, $\varphi_{b_2} = x^{l_o}_{\Delta, f} > x_{ss}(d_{\varphi}, v^{l_o}_{f}, T)$, $\varphi_{b_3} = x^{l_o}_{\Delta, p} > x_{ss}(d_{\varphi}, v, T)$, and the $l_o$ superscript denotes vehicles in the other lane relative to the ACC vehicle and the $_p$ and $_f$ subscripts denote the vehicles preceding and following the ACC vehicle. We use the following signal-to-property assignments for the inequality properties in $\varphi_b$:  $\propsignals{\varphi_{b_1}} = \lbrace v, x_{\Delta} \rbrace$,  $\propsignals{\varphi_{b_2}} = \lbrace x^{l_o}_{\Delta, f} \rbrace$,  $\propsignals{\varphi_{b_3}} = \lbrace x^{l_o}_{\Delta, p}, v \rbrace$. We note that $\varphi_{b_2}$ depends on $v^{l_o}_{f}$ which is not directly measurable and must be estimated. Thus the $\rhoettwc{\cdot}$ policy cannot be used for $\varphi_b$.
Additionally, we construct an identical property except we replace the $\lor$ operator in $\varphi_b$ with an $\land$ for monitoring. This is to demonstrate what would happen if the ETT relaxation mechanism for the $\lor$ operator in Definition \ref{def:prop-operators-ett-regulation} was not implemented and we treated the $\lor$ operators as an $\land$. We consider two scenarios: 1) a critical scenario where there is little but enough room to overtake at the right time, and 2) a non-critical scenario where there is plenty of room to overtake.

\paragraph*{Parameter identification}
Based on the results in the single-lane scenario, we know that $T_s = 0.01 s$ is necessary for the TT policy to ensure positive robustness in the most critical scenario. In the case of the multilane scenario, the most critical case is identical to that of the one-lane scenario when there is no room to safely switch to the fast lane. Thus, we conclude that the same sampling frequency is necessary for the multilane scenario. 

Next, we consider the CETT and $\rho$ETT transmission strategies. If there is not enough room to overtake in the fast lane, the ETTs should still ensure positive robustness in the current lane. Thus, the ETTs for the sensors that monitor the speed of the ACC vehicle and the distance between the ACC vehicle and the preceding vehicles in either lane, should use identical parameters to those found for the $\varphi_a$ safety property in the previous section. However, the property $\varphi_{b_2}$ in $\varphi_b$ is a different case, as the estimated speed of the following vehicle is used to determine the safe distance rather than the speed of the ACC vehicle. To determine safe ETT parameters for the sensor measuring the distance to the following vehicle in the fast lane ($x^{l_o}_{\Delta, f}$), we initially consider the critical scenario to determine safe parameters. We test a range of different $\rho$ETT policy parameters for the $x^{l_o}_{\Delta, f}$ signal and run 20 simulations for each configuration.

In Table \ref{tab:results}, we note the minimum robustness and corresponding number of sent measurements in the critical scenario ($\varphi_b$) for the best tested configuration for the CETT, $\rho$ETT, TT and the $\rho$ETT regulation policy with the $\lor$ replaced by an $\land$ operator. The robustness is measured as if the $\lor$ operator were present. The results show that the $\rho$ETT policy with the $\lor$ operator triggers \textbf{89.1\%} and \textbf{96.6\%} fewer events compared to the CETT and $\rho$ETT without the $\lor$ operator, respectively. Another notable result is that the $\rho$ETT without the $\lor$ operator transmits \textbf{3.17} times more packets than the CETT approach. Upon investigation, we find that many of the transmitted measurements can be attributed to the vehicles in the fast lane as the robustness of the properties $\varphi_{b_2}$ and $\varphi_{b_3}$ is negative for a significant duration of the simulation. This results in the corresponding ETTs being zero for the duration of the negative robustness resulting in many transmitted packets. However, the CETT does not suffer from this problem as the ETTs never reach zero. Based on the results in Fig. \ref{fig:rho-m-tradeoff}, we see that once we reach a certain minimum robustness, many additional transmitted measurements will not significantly increase $\rho_{min}$. Along with the fact that measurement noise is typically present, this may suggest that some minimum ETT $ > 0$ may be a good idea to decrease the number of unnecessary transmitted measurements in the case of low robustness.

Finally, we consider a non-critical multilane scenario where there is plenty of space to overtake the braking vehicle, resulting in an overall larger robustness. The parameters found for the critical case are reused and the results are noted again in Table \ref{tab:results}. The results for the non-critical multilane scenario show similar results as the critical scenario where the $\rho$ETT policy transmits \textbf{96.3 \%} and \textbf{97.1\%} fewer events compared to the CETT and $\rho$ETT policy without the $\lor$ operator respectively. The $\rhoett{\cdot}$ policy triggers significantly fewer events in the non-critical scenario compared to the critical scenario while the number of triggered events for the CETT policy stays almost the same. This result highlights the ability of the $\rhoett{\cdot}$ mechanism to adapt based on critical and non-critical situations.

\section{Discussion}\label{sec:discussion}
We now discuss our proposed ETT regulation method in light of the obtained numerical results as well as differences and similarities with other approaches in SoTA.

\paragraph*{Parameter Tuning and Formal Guarantees} Section~\ref{sec:case-study} also demonstrates the parameter tuning procedure used to select constant ETTs or ETT regulation parameters. In the simulated scenarios, the parameters chosen ensure safety. However, these parameters do not provide formal guarantees for a class of systems or scenarios, i.e., they do not provide a guarantee for potentially similar scenarios.  
For example, if we consider the multilane scenario and the sensor measuring $x^{l_o}_{\Delta, f}$, the parameters resulting in triggering the smallest number of events were sufficient to ensure a positive value of $\rho_{min}$. This is likely due to added acceleration meaning that the estimated robustness of $\varphi_{b_3}$ is typically larger than the estimated robustness. However, if the drift was in the opposite direction, then the large ETT (or small $\epsilonparam{y_i}{\varphi}$ value) may not be sufficient to satisfy the requirement. As a result, when using such a numerical parameter tuning method, one has to consider what situations have to be tested to ensure that the chosen parameters ensure safety in all situations that the system will experience, which is beyond the scope of this paper. 

A potential way to combat the safety issue described above is to make the control behavior aware of the uncertainty introduced by the event-triggering threshold. For example, consider the scenario where the ACC vehicle is far away from the safe distance to a preceding vehicle in the slow lane but for some reason desires to change to the fast lane. The property $\varphi_b$ will then provide high ETTs to the sensors measuring the distance to the vehicles in the fast lane. If the ACC vehicle for some reason desires to switch lanes, it will have an inaccurate estimate of the vehicles in the other lane, which may not guarantee a safe overtaking. Thus, the uncertainty information should be taken into account in the control behavior to ensure safe operation. The reason this is not an issue in the case study, is that the overtaking behavior is designed to only occur when the ACC vehicle is close to the lead vehicle in the slow lane. In this scenario, the low distance to the lead vehicle will not contribute to altering the ETTs related to the distance of the other vehicles in the fast lane through the $\lor$ operator.

To provide formal guarantees, other comparable methods, such as the co-design of a controller and event-triggering mechanism that ensures the stability of a networked control system \cite{zhangSurveyRecentAdvances2016, zhangNetworkedControlSystems2020a} have been proposed.
Such formalisms rely on the ability to formalize a stability condition, which our method does not require as the safety condition is implicitly specified in the optimization problem. Additionally, our method allows us to simultaneously consider performance properties. Allthough this was not shown in our case study, the overall procedure and concept are the same. 

\paragraph*{Similarities and Differences with SoTA} Using propositional properties to describe system safety/performance requirements has similarities with other ways of describing systems, such as Fuzzy modeling and control \cite{yagerEssentialsFuzzyModeling1994}. Fuzzy modeling and control allows the specification of IF-THEN statements to describe control and system behavior which is conceptually similar to how the $\rightarrow$ operator can be used. Several works have been published for event-triggered control of fuzzy systems \cite{panEventtriggeredFuzzyControl2017, suEventtriggeredFuzzyControl2018, panSecurityBasedFuzzyControl2022} but they typically only consider behaviors of the form $(\varphi_1 \land \varphi_2 ...) \rightarrow \text{behavior}$. Additionally, many of the formal methods that guarantee safety, do not explore the tradeoff between varying the ETTs for different sensors and its impact on the overall number of transmitted packets.  

\paragraph*{Connection to Signal Temporal Logic} As mentioned in Section~\ref{sec:stl}, PL is a subset of STL, which extends propositional logic with temporal operators. For example, the \textit{eventually} ($\diamond_I$) and \textit{always} ($\square_I$) operators allow the system to verify that something will eventually become true or is always true within some interval $I$, respectively. This enables STL to specify properties over constrained time intervals rather than simply stating that something always has to be true. Several results on controller synthesis to ensure satisfaction of STL properties or other similar formalisms have been published (some useful references include \cite{beltaFormalMethodsControl2019, ramanReactiveSynthesisSignal2015}) but only very few on event-triggered control \cite{lindemannEventtriggeredFeedbackControl2018, longDynamicEventTriggered2024, maityEventTriggeredControllerSynthesis2018} exist. Furthermore, most of the event-triggered control publications only consider a subset of STL excluding the $\lor$ operator. To our knowledge, the publication that is closest to our work is \cite{maityEventTriggeredControllerSynthesis2018}, which develops an event-triggering strategy that ensures that the event-triggered state trajectory stays within an $\epsilon$-tube of the ideal trajectory given by the continuous state feedback controller for Temporal Logic over Reals (RTL) \cite{reynoldsComplexityTemporalLogic2010} formulas. These correspond roughly to unbounded PL formulas but without the definition of numerical satisfaction semantics. Additionally, the event-triggering condition in \cite{maityEventTriggeredControllerSynthesis2018} is evaluated based on the full state and only considers the option of transmitting the full state, rather than measurements of individual signals. 

Another similarity with the SoTA is that the optimization problems presented in this paper are always constrained to ensure that the robustness of the property is positive for all time steps. This corresponds to an implicit unbounded STL $\square$ operator. The subset of STL considered in \cite{lindemannEventtriggeredFeedbackControl2018,longDynamicEventTriggered2024}, considers bounded $\square$ operators but does not consider disjunctions. Extending our approach to include the STL temporal operators and refining the currently proposed ETT regulation mechanisms will be a core part of our future work.

\paragraph*{Assumptions and Impact in Real Systems} The case study makes several other assumptions that need to be handled in a real-world setting. For example, we assumed that the necessary information for evaluating the event-triggering condition is available at the sensor without the need for communication. In a real-world scenario, this assumption does not hold in every case and the required communication could diminish the savings introduced by the ETT. The SOD update error may be a better fit in a real-world scenario as it does not require additional information to predict the value of the measurement. We also note that the ETT refinement as a result of the $\lor$ operator will likely also require additional communication since properties (and thus signals) can affect the ETT of other signals.
However, the fact that there is a large gap between the $\rho$ETT and the constant ETT (namely, the $\rho$ETT approach triggers between \textbf{41.8 - 96.3 \%} less events compared to the CETT), provides ample room to address these challenges while maintaining significant gains. 
Despite many of these potential challenges, the expressiveness of PL has the potential to save communication resources and optimize performance in a variety of scenarios by enabling fine-grained control of the ETT through our proposed method.

\section{Conclusions and Future Work}\label{sec:conclusion}

In this paper, we have presented a novel ETT regulation mechanism that infuses ETT regulation with propositional logic. Through the expressiveness of PL, our approach enables accurately defining circumstances in a wide array of problems under which the ETTs are either increased or decreased depending on the safety and/or performance requirements of the system as well as how well these requirements are satisfied at runtime. We explored the intuition and possibilities of our method using a safety-related case study. The case study showed a large potential for reduction in the number of triggered events, e.g. between \textbf{41.8 - 96.3 \%} fewer events while maintaining similar minimum safety. Assessing the usability of our proposed method in a wide range of systems will be a core part of our future work as this will help us identify potential issues and beneficial improvements.
Future work will focus on expanding our approach to include temporal operators from Signal Temporal Logic (STL) in order to capture more complex properties that have specific timing considerations as well as methods for parameter tuning of the system.

\appendix
\section{Interval arithmetic properties}
\subsection{Proof of Lemma \ref{lem:interval-reversal}}\label{apdx:interval-reversal-proof}
\begin{proof}
    The scalar $\alpha$ can be represented as the degenerate interval $A = [\alpha, \alpha]$. Multiplying $X$ and $A$, according to \cref{eq:interval-multiplication}, we have: $A \cdot X = [\min S, \max S] \text{ where } S = \lbrace \underline{AX}, \underline{A}\overline{X}, \overline{A}\underline{X}, \overline{AX}\rbrace$. Since $\overline{A} = \underline{A}$, then $\underline{AX} = \overline{A}\underline{X}$ and $\overline{AX} = \underline{A}\overline{X}$. Since $\alpha < 0$ and $\overline{X} > \underline{X}$, consequently $\alpha \overline{X} < \alpha \underline{X}$, resulting in $\min S = \alpha \overline{X}$ and $\max S = \alpha \underline{X}$ and the interval $A \cdot X = [\alpha \overline{X}, \alpha \underline{X}]$.
\end{proof}
\subsection{Proof of Lemma \ref{lem:ett-interval-width}}\label{apdx:ett-interval-width-proof}
\begin{proof}
    We show the proof for the case $\alpha_X > 0 \text{ and } \alpha_Y < 0$. The proof for the remaining cases ($\alpha_X > 0 \text{ and } \alpha_Y > 0, \alpha_X < 0 \text{ and } \alpha_Y > 0, \alpha_X < 0 \text{ and } \alpha_Y < 0$) can be obtained by following the same approach.
    Using the result of Lemma \ref{lem:interval-reversal}:
    \small
    \begin{align*}
        & w(\alpha_XX(x, \delta_x) + \alpha_YY(y, \delta_y)) \\
        & = w([\alpha_X(x - \delta_x), \alpha_X(x + \delta_x)] + [-|\alpha_Y|(y + \delta_y), -|\alpha_Y|(y - \delta_y)]) \\
        & = w([\alpha_X(x - \delta_x) + -|\alpha_Y|(y + \delta_y), \alpha_X(x + \delta_x) + -|\alpha_Y|(y - \delta_y)]) \\
        & = \alpha_X(x + \delta_x) + -|\alpha_Y|(y - \delta_y) - (\alpha_X(x - \delta_x) + -|\alpha_Y|(y + \delta_y)) \\
        & = 2 \alpha_X\delta_x + 2 |\alpha_Y|\delta_y = 2 |\alpha_X|\delta_x + 2|\alpha_Y|\delta_y.
    \end{align*}
    \normalsize
\end{proof}

\bibliographystyle{IEEEtran}
\bibliography{lib.bib}

\end{document}